\documentclass[runningheads]{llncs}

\usepackage{amsmath,amssymb,amsfonts, mathtools}
\usepackage{algorithm, algpseudocode}
\usepackage{graphicx}
\usepackage{textcomp}
\usepackage{xcolor}
\usepackage{booktabs} 
\usepackage{hyperref}
\usepackage{amsmath}
\usepackage{enumitem}
\usepackage{relsize}
\usepackage{subcaption}
\usepackage{multirow}
\usepackage[depth=3]{bookmark}

\newtheorem{observation}{Observation}

\DeclarePairedDelimiter{\ceil}{\lceil}{\rceil}

\begin{document}
\title{Improving Locality Sensitive Hashing by Efficiently Finding Projected Nearest Neighbors}

\author{Omid Jafari\orcidID{0000-0003-3422-2755} \and
Parth Nagarkar\orcidID{0000-0001-6284-9251} \and
Jonathan Monta\~no\orcidID{0000-0002-5266-1615}
}
\authorrunning{O. Jafari et al.}

\institute{New Mexico State University, Las Cruces, US \and
\email{\{ojafari, nagarkar, jmon\}@nmsu.edu}}
\maketitle              
\begin{abstract}
Similarity search in high-dimensional spaces is an important task for many multimedia applications. Due to the notorious \textit{curse of dimensionality}, approximate nearest neighbor techniques are preferred over exact searching techniques since they can return \textit{good enough} results at a much better speed. \textit{Locality Sensitive Hashing} (LSH) is a very popular random hashing technique for finding approximate nearest neighbors. Existing state-of-the-art Locality Sensitive Hashing techniques that focus on improving performance of the overall process, mainly focus on minimizing the total number of IOs while sacrificing the overall processing time. The main time-consuming process in LSH techniques is the process of finding neighboring points in projected spaces. 
We present a novel index structure called \underline{r}adius-\underline{o}ptimized \underline{L}ocality \underline{S}ensitive \underline{H}ashing (\textit{roLSH}). With the help of sampling techniques and Neural Networks, we present two techniques to find neighboring points in projected spaces efficiently, without sacrificing the accuracy of the results. Our extensive experimental analysis on real datasets shows the performance benefit of \textit{roLSH} over existing state-of-the-art LSH techniques. 

\end{abstract}

\section{Introduction}
Finding nearest neighbors is an important problem in many domains such as information retrieval, computer vision, machine learning, multimedia retrieval, etc. For low-dimensions ($<10$), popular tree-based index structures, such as KD-tree
, Quad-tree
, etc. are effective, but for higher number of dimensions, these index structures suffer from the well-known problem, \textit{curse of dimensionality} (where the performance of these index structures is often out-performed even by linear scans) \cite{Chavez:2001:SMS:502807.502808}. One solution to this problem is to search for approximate results instead of exact results. In many applications where strictly correct results are not necessary, approximate results can produce \textit{good enough} results while achieving much better running times. The goal of the $c$-approximate version of the Nearest Neighbor problem (ANN) is to find nearest neighbors for a given query point that are within $c*R$ distance (where $c>1$). 

\subsection{Locality Sensitive Hashing}
\label{sec:lsh}
Locality Sensitive Hashing (LSH) \cite{Gionis:1999:SSH:645925.671516} is a very popular technique for solving the Approximate Nearest Neighbor problem in high-dimensional spaces. LSH uses \textit{random} projections to map high-dimensional points to lower dimensional representations. The intuition behind LSH is that nearby points in high-dimensional spaces will map to same (or nearby) hash buckets in the projected lower dimensional space with a high probability (and vice-versa). Since the original LSH index structure was proposed for Hamming distance, LSH families have been proposed for other popular distances such as the Euclidean distance \cite{Datar:2004:LHS:997817.997857}. 
The main benefits of LSH are three-fold: 1) LSH provides theoretical guarantees on the accuracy of the results, 2) LSH can answer ANN queries in sub-linear time with respect to the dataset size, and 3) LSH can be easily implemented as external memory-based index structures, thus making them more scalable \cite{Liu:2019}. While the original LSH design suffered from large index sizes \cite{Lv:2007:MLE:1325851.1325958}, recent works \cite{Gan:2012:LHS:2213836.2213898,Huang:2015:QLH:2850469.2850470,Tobias:2019} have either improved theoretical bounds or introduced techniques such as \textit{Collision Counting} (Section \ref{sec:prelim}) to reduce the number of required hash functions. Due to the popularity of LSH in diverse applications \cite{Earthquake:10.14778/3236187.3236214,yang2004hierarchical}, several research works have been proposed to improve the search efficiency and/or accuracy of LSH techniques \cite{Lv:2007:MLE:1325851.1325958,Gan:2012:LHS:2213836.2213898,Huang:2015:QLH:2850469.2850470,Tobias:2019,Jafari:2019:QCI:3323873.3325048,Liu:2014:SEI:2732939.2732947,Liu:2019,PM-LSH/3377369.3377374}.

\subsection{Motivation of our work: Improving the Efficiency of Existing State-of-the-Art LSH Techniques}
\label{sec:ourMotiv}
One of the important benefits of LSH is their ease of implementation as external storage based algorithms. State-of-the-art external memory-based algorithms (namely C2LSH \cite{Gan:2012:LHS:2213836.2213898}, QALSH \cite{Huang:2015:QLH:2850469.2850470}, and I-LSH \cite{Liu:2019}) use a bucket-expansion strategy to find points from neighboring buckets. C2LSH and QALSH use a bucket exponential expansion strategy, whereas I-LSH uses an incremental expansion strategy. While I-LSH is the state-of-the-art algorithm that minimizes disk I/Os, it achieves this optimization at the expense of a costly overall processing time as shown in Section \ref{sec:exp}.\footnote{There is no existing work that compares the overall performance (in terms of query processing time, disk I/Os, and accuracy) of C2LSH, QALSH, and I-LSH. We present a detailed performance analysis between these works as a technical report (\url{https://3m.nmsu.edu/lsh-survey/}), which is also under submission at SISAP 2020.} Additionally, random I/Os (disk seeks) are known to be bottleneck in query processing \cite{10.1145/3209900.3209903} and much more expensive than sequential I/Os \cite{10.1007/s00778-017-0480-7}. I-LSH reduces overall I/Os by mainly reducing sequential I/Os. In this paper, our goal is to design an LSH external memory technique, \textit{roLSH}, that can reduce overall IOs, \textit{mainly random I/Os}, (by finding neighboring points efficiently) which improves the overall query processing time. 

\subsection{Contributions of this Paper}

In this paper, we propose a novel approach, called \underline{r}adius-\underline{o}ptimized \underline{L}ocality \underline{S}ensitive \underline{H}ashing (\textit{roLSH}) for efficiently finding top-k approximate nearest neighbors in high-dimensional spaces. Our main contributions are as follows:

\begin{itemize}
	\item We present a sampling-based technique, \textit{roLSH-samp}, that reduces the overall random disk I/Os which improves the query processing time while satisfying the theoretical guarantees of LSH. We provide the theoretical analysis for the correctness of \textit{roLSH-samp}.
	\item We further improve the efficiency by proposing a Neural Network-based technique, \textit{roLSH-NN}, for an improved prediction of projected radiuses (and thus further reduction in random disk I/Os), and hence further improving the performance without affecting the query accuracy. To the best of our knowledge, we are the first work to improve LSH parameters by using Neural Networks.
	\item Lastly, we experimentally evaluate both techniques of \textit{roLSH} on real high-dimensional datasets and show that \textit{roLSH} can outperform the state-of-the-art solutions in terms of performance while providing similar query accuracy.
\end{itemize}

\section{Related Work}
\label{sec:relWork}

Locality Sensitive Hashing is a popular technique for solving the Approximate Nearest Neighbor (ANN) problem in high-dimensional spaces. It was first introduced in \cite{Gionis:1999:SSH:645925.671516} for the Hamming distance and later extended to the Euclidean distance (E2LSH) \cite{Datar:2004:LHS:997817.997857}. These structures suffered from large index sizes due to the need to have large number of hash functions in multiple hash tables \cite{Gan:2012:LHS:2213836.2213898}. Additionally, a \textit{magic radius} need to be inputted to find the neighboring projected points, and in order to find the desired number of results, this \textit{magic radius} was arbitrarily chosen to be very high. Multi-Probe LSH \cite{Lv:2007:MLE:1325851.1325958} presented a technique to probe neighboring buckets if enough number of results were not found. C2LSH \cite{Gan:2012:LHS:2213836.2213898} introduced a \textit{Collision Counting} approach that reduced the need to have multiple hash tables, and hence reduced the overall index size. SK-LSH \cite{Liu:2014:SEI:2732939.2732947} introduced a linear ordering on the disk pages with the help of Z-order curve in order to reduce the overall I/Os. The drawback of SK-LSH was that it was created on the original LSH design, and hence also suffered from the \textit{magic radius} problem. QALSH \cite{Huang:2015:QLH:2850469.2850470} introduced query-aware hash functions and further reduced the number of hash functions necessary to achieve theoretical guarantees. The work closest to our proposed idea is I-LSH \cite{Liu:2019}, which introduces an incremental strategy for finding nearest neighbors in the projected space, as explained in the next section (\ref{sec:existing}). Recently, PM-LSH \cite{PM-LSH/3377369.3377374} developed a novel tunable confidence interval while using a PM-tree to solve c-ANN queries.\footnote{PM-LSH code is not yet released.}

\subsection{Existing Techniques for finding Neighboring Projected Points}
\label{sec:existing}
C2LSH \cite{Gan:2012:LHS:2213836.2213898} also introduced the concept of \textit{Virtual Rehashing}. The goal of \textit{Virtual Rehashing} is to find neighboring points that collide in neighboring hash buckets. The naive solution to finding neighboring points is to use a large projected radius such that enough neighboring points are found to return top-k results. The projected radius is entirely dependent on the data distribution, and as we show in Figure \ref{fig:radHists}, these projected radiuses can vary significantly. Hence, using an arbitrarily large radius results in wasted I/Os and unnecessary processing. Instead, \textit{Virtual Rehashing} starts with a very small radius (R=1), and then exponentially increases the radius in the following sequence: $R = 1, c, c^2, c^3...$. If at \textit{level-R}, enough candidates are not found, the radius is increased until enough query results are found. C2LSH \cite{Gan:2012:LHS:2213836.2213898} and QALSH \cite{Huang:2015:QLH:2850469.2850470} follow this exponential expansion strategy. I-LSH \cite{Liu:2019} introduces an incremental strategy where, instead of expanding the search radius exponentially, they find the nearest point to the query in each projection. C2LSH and QALSH stores points in hash buckets which are stored in disk pages. I-LSH stores each data point separately and hence instead of reading disk pages that stores a group of data points, only reads a point (which effectively is the same as reading a disk page of 4 bytes). While they save on disk I/O operations by this method, this is a very costly operation (as shown in Section \ref{sec:exp}) since this process has to be done thousands of times for larger radiuses.

\section{Background and Key Concepts}
\label{sec:prelim}

In this section, we describe the key concepts behind LSH. We mainly use the notations and formulations described in the seminal paper on Euclidean LSH families \cite{Datar:2004:LHS:997817.997857} and C2LSH \cite{Gan:2012:LHS:2213836.2213898}. 

\noindent\textbf{Hash Functions: }A hash function family $H$ is ($R$, $cR$, $p_1$, $p_2$)-sensitive if it satisfies the following conditions for any two points $x$ and $y$
in a $d$-dimensional dataset $D \subset \mathbb{R}^d$: if $|x - y| \leq R$, then $Pr[h(x) = h(y)] \geq p_1$, and if $|x - y| > cR$, then $Pr[h(x) = h(y)] \leq p_2$.

Here, $p_1$ and $p_2$ are probabilities and $c$ is an approximation ratio. In order for LSH to work, $c > 1$ and $p_1 > p_2$. 
The above definition states that the two points $x$ and $y$ are hashed to the same bucket with a very high probability $\geq p_1$ if they are close to each other (i.e. the distance between the two points is less than or equal to $R$), and if they are
not close to each other (i.e. the distance between the two points is greater than $cR$), then they will be hashed to the same bucket with a low probability $\leq p_2$. 
In the original LSH scheme for Euclidean distance, each hash function is defined as $h_{\vec{a},b} (x) = \left\lfloor{\frac{\vec{a}.x + b}{w}}\right\rfloor,$ where $\vec{a}$ is a $d$-dimensional random vector with entries chosen
independently from the standard normal distribution $N(0,1)$ and $b$ is a real number chosen uniformly from $[0, w)$, such that $w$ is the width of the hash bucket \cite{Datar:2004:LHS:997817.997857}.  
This leads to the following collision probability function~\cite{Datar:2004:LHS:997817.997857}, which states that if $||x,y||=r$, then the probability that $x$ and $y$ map to the same hash bucket for a given hash function $h_{\vec{a},b} (x)$ is:
$P(r) = \int_{0}^{w}{\frac{1}{r} \frac{2}{\sqrt{2\pi}} e^{\frac{-t^2}{2r^2}}(1-\frac{t}{w}) dt}$. 
Here, the collision probability $P(r)$ is decreasing on $r$ for a given $w$. 

For a $t$, which is the largest absolute value of a coordinate of point in $D$, and for every $b$ uniformly drawn from the interval $[0,c^{\lceil\log_c td\rceil}w^2]$ and $R=c^n$ for some $n\leq \lceil\log_c td\rceil$ we have that 
$h^R(x)=\left\lfloor\frac{h_{\vec{a},b} (x)}{R}\right\rfloor$
is $(R,cR,p_1,p_2)$-sensitive, where $p_1=p(1)$ and $p_2=p(c)$ \cite{Gan:2012:LHS:2213836.2213898}.

\noindent\textbf{Collision Counting:} In \cite{Gan:2012:LHS:2213836.2213898}, authors theoretically show that two close points $x$ and $y$ collide in at least $l$ hash layers with a probability $1-\delta$, when the total number,
$m$, of hash layers are equal to:
$m = \ceil[\big]{\frac{\ln(\frac{1}{\delta})}{2(p_1-p_2)^2}(1+z)^2}$.
Here, $z=\sqrt{\ln({\frac{2}{\beta}})/\ln({\frac{1}{\delta}}})$, where $\beta$
is the allowed false positive percentage (i.e. the allowed number of points whose distance with a query point is greater than $cR$). C2LSH sets $\beta=\frac{100}{n}$, where $n$ is the cardinality of the dataset. Further, only those points that collide at least $l$ times, 
where $l$ is the collision count threshold, which is calculated as following: $l = \lceil\alpha \times m\rceil$, where the collision threshold percentage, $\alpha$, is 
$\alpha = \frac{zp_1 + p_2}{1+z}$.
Since C2LSH creates only one hash function per hash layer, the number of hash functions are equal to the number of hash layers.  

\section{Problem Specification}
 
The approximate version of the nearest neighbor problem, also called \textit{c-$\delta$-approx-imate Nearest Neighbor search}, aims to return points that are within $c * R$ distance from the query point with probability at least $1-\delta$, where $c>1$ is a user-defined approximation ratio, $R$ is the distance of the query point from its nearest neighbor, and $\delta$ is a user-defined error probability. 

In this paper, our goal is to return c-$\delta$-ANNs for a given query $q$ while reducing the overall processing time and satisfying the theoretical guarantees. In Section \ref{sec:rolsh}, we present the processing cost breakdown of the LSH process based on which we design our proposed index structure, radius-optimized Locality Sensitive Hashing (\textit{roLSH}).

\begin{figure*}[h]
	\centering
	\begin{subfigure}[b]{0.31\textwidth}
		\centering
		{\includegraphics[width=\linewidth]{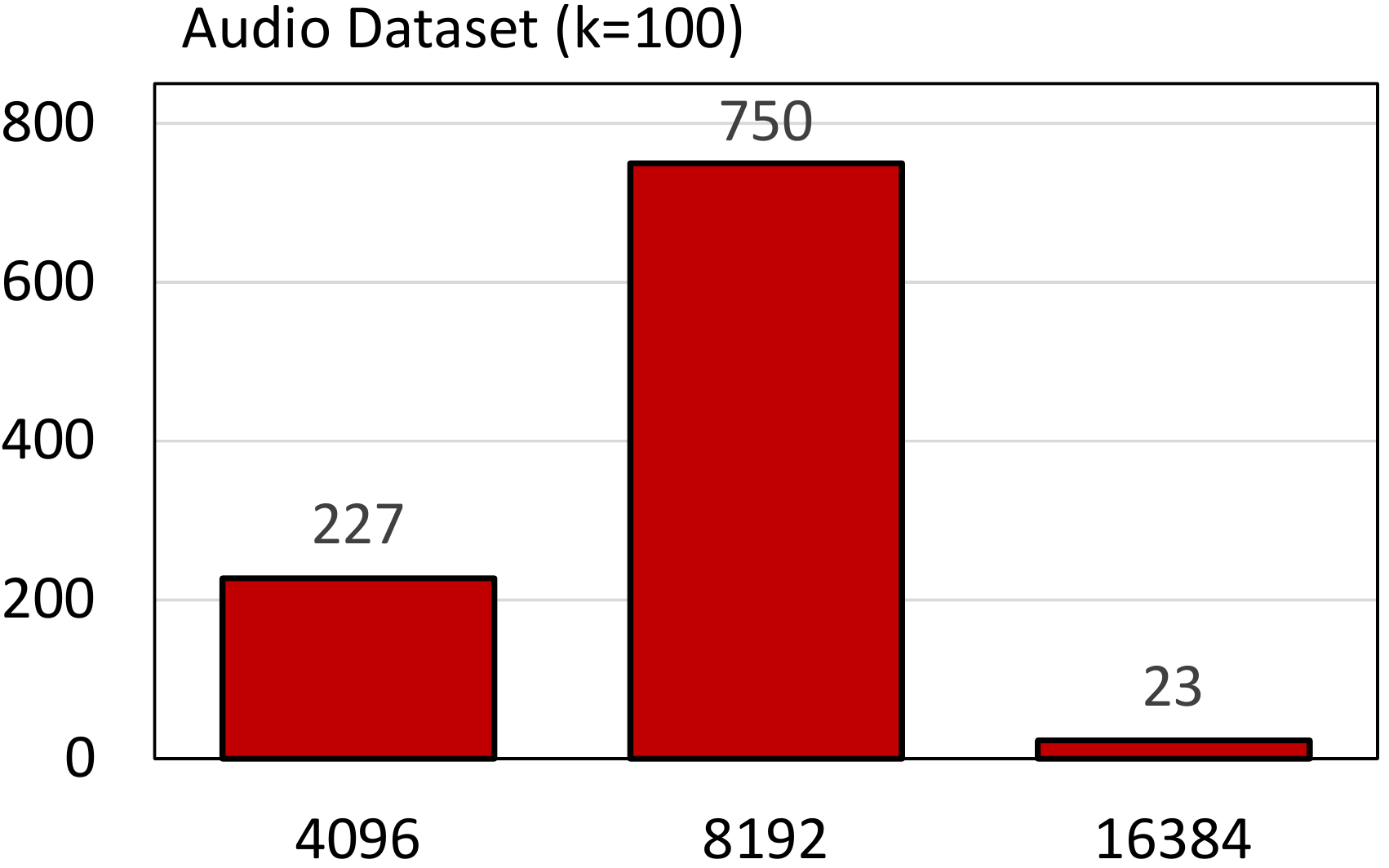}}
	\end{subfigure}\quad
	\begin{subfigure}[b]{0.31\textwidth}
		\centering
		{{\includegraphics[width=\linewidth]{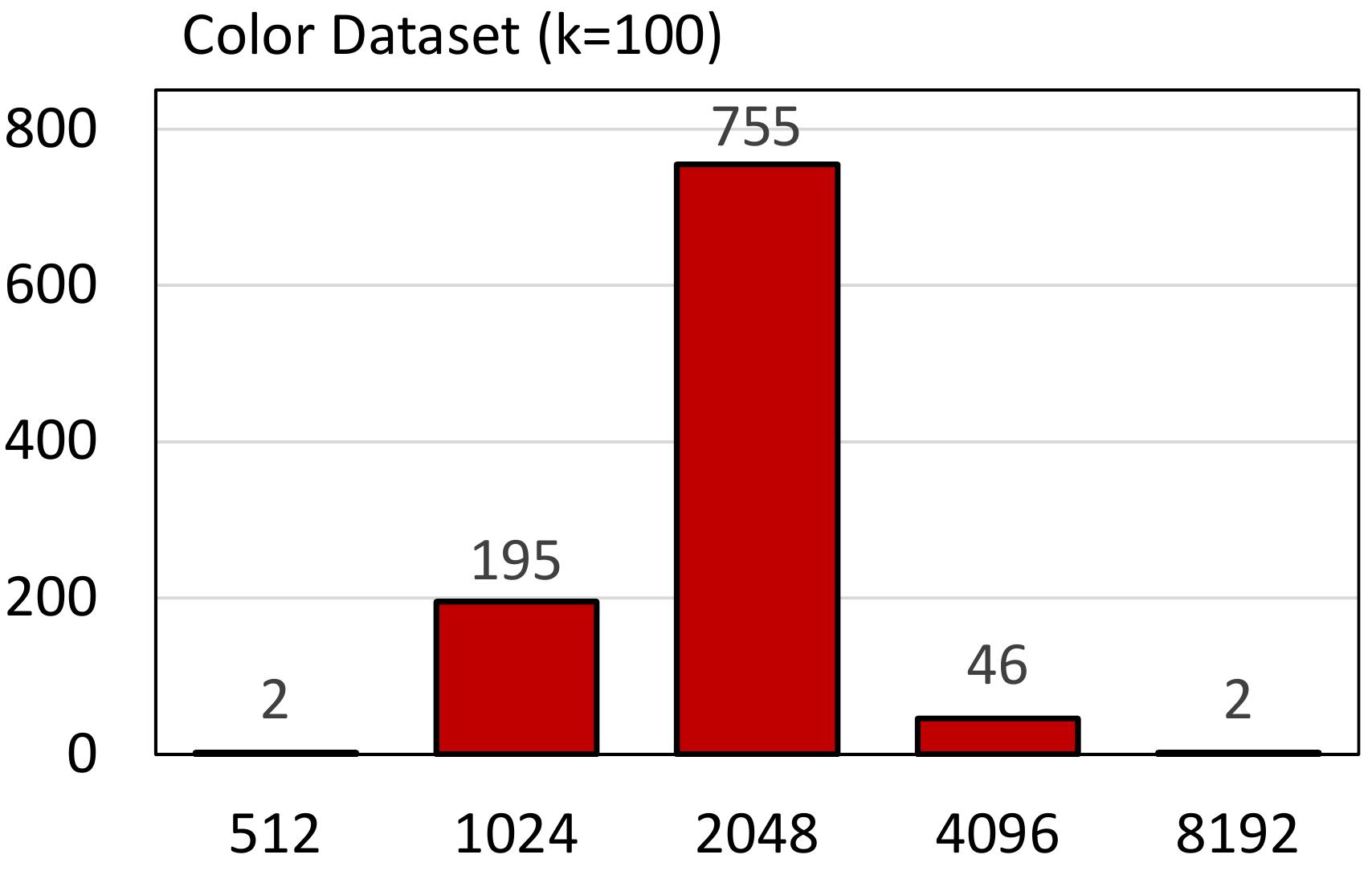}}}
	\end{subfigure}\quad
	\begin{subfigure}[b]{0.31\textwidth}
		\centering
		{{\includegraphics[width=\linewidth]{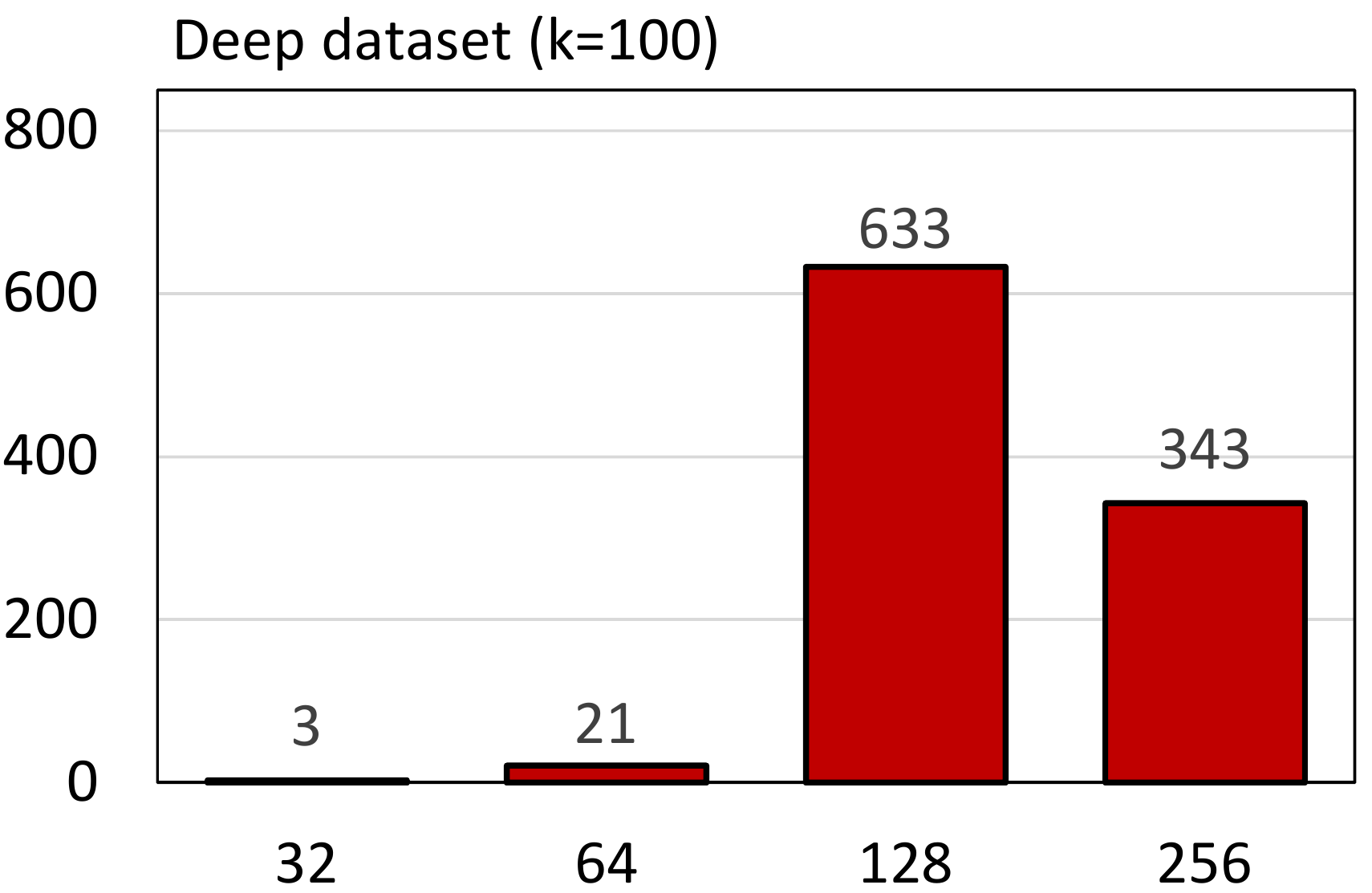}}}
	\end{subfigure}
	\caption{Frequency (Y-axis) of Final Radius Values (X-axis) for finding Top-100 Points for 1000 Point Queries on Different Datasets using C2LSH}
	\label{fig:diffSampling}
	\vspace*{-0.4in}
\end{figure*}

\section{roLSH}
\label{sec:rolsh}

In this section, we present the design of \textit{roLSH}, which consists of two strategies for efficiently finding neighboring points in the hash functions. We introduce and describe these two strategies in this section: a sampling-based strategy, called \textit{roLSH-samp}, and a Neural Network-based strategy, called \textit{roLSH-NN}. 

\subsection{Sampling-based Improved Virtual Rehashing Strategy}
\label{sec:roLSH-samp}

In Section \ref{sec:relWork}, we explained the original Virtual Rehashing strategy (denoted as \textit{oVR} strategy) as proposed in C2LSH \cite{Gan:2012:LHS:2213836.2213898}. The initial radius is set to 1, and if sufficient results are not found, then the radius is increased in an exponential sequence: $R = 1, c, c^2, c^3...$ until sufficient number of results are found. The main drawback of this approach is when the values of $R$ become larger (i.e. when the difference between two consecutive radius values is large - e.g. 4096 and 8192). In such situations, it happens frequently that very few (or no) nearest neighbor points are found at radius value 4096 but all (and lot more) are found at radius 8192. Thus, for example, if the actual radius of the $k$th-nearest point was near 5000, then index files corresponding to radius 5000-8192 will be read unnecessarily from the disk, leading to expensive wasted IO operations. Instead, we propose a sampling-based improved Virtual Rehashing strategy (denoted as \textit{roLSH-samp}) based on the following observation: 
\begin{observation}
	For high-dimensional datasets, the required radius values for a $k$ value are similar to each other for different query points for a given dataset.
\end{observation}
This observation was also noted by a very recent paper \cite{PM-LSH/3377369.3377374} where the authors show that the homogeneity of the distance distributions of data points in different high-dimensional datasets is very high. Figure \ref{fig:diffSampling} shows our observation on popular real high-dimensionsal datasets with varying cardinalities and dimensionalities (Audio \cite{Audio}, Color \cite{Color}, 
Deep \cite{babenko2016efficient}). For 1000 randomly chosen query points, we report the final radius values (using the Virtual Rehashing technique from C2LSH \cite{Gan:2012:LHS:2213836.2213898}) for top-100 points. By leveraging the above stated simple observation, we design an improved, simple, and effective \textit{Virtual Rehashing} technique: we execute a sample set of randomly chosen queries for a given $k$ and count the number of occurrences of the final radius value. We choose our initial radius value that is before the radius with the maximum count of sampled queries. E.g. in the Audio dataset (Figure \ref{fig:diffSampling}), the radius with the maximum count is 8192. For these queries, it means that the optimal radius would be between 4096 and 8192. Hence we choose our initial radius value to be 4096.  
Thus instead of starting at the initial radius of 1, we find an \underline{i}mproved \underline{i}nitial starting \underline{r}adius (denoted as $i2R$) based on sampling queries at the end of the indexing process. Note that, since this is done during the indexing phase, it has no overhead during query execution. Additionally, we do not need to store any distance pairs, but simply need to execute a small number of top-$k$ point queries to find the initial starting radius. Once the initial starting radius ($i2R$) is found, we leverage the same exponential sequence strategy as C2LSH, such that

\[ R = 
	\begin{cases} 
		i2R + 2 ^ x & 0 \leq x \leq \log_2 i2R \\
		2 ^ x & x > \log_2 i2R 
	\end{cases}
\]

Thus, for 1000 random queries on the Audio dataset, using the original \textit{Virtual Rehashing} technique, the average final radius is 7450 for $c=2$ and $k=100$. On the other hand, using our improved strategy, the average final radius is 6083, which leads to significant savings in the IO.

Note that, one disadvantage of this approach is that there potentially can be queries that finish with a radius value much lower than the chosen initial radius. E.g., in the Color dataset (Figure \ref{fig:diffSampling}), our strategy will choose $i2R = 1024$. As you can see, there were 2 out of 1000 queries whose final radius value to find top-100 points was 512. In this case, the \textit{iVR} strategy will do wasted work by starting (and ending) at 1024. 

\begin{lemma}\label{samplingLemma}
For those queries whose required radius in \textit{oVR}  is at least $(2\times i2R)$,  \textit{iVR} strategy will generate less IOs than the \textit{oVR} strategy. 
\end{lemma}	
\begin{proof}
	Set $R=i2R$. By construction of the sequence of radii in  \textit{oVR}, it is enough to assume  that the required radius   is $2R$, that is, the actual radius $r$ of the $k$th-nearest point satisfies $R< r\leq 2R$. In the \textit{oVR}, the sequence of radii needed to find the  $k$th-nearest point has $\log_2 R + 2$ elements, that is, $1, 2, 4,\ldots, 2R$. On the other hand, for the same query $q$,  \textit{iVR}   analyzes at most $\log_2 R + 1$   radii, that is, $R+1, R+2, \ldots, 2R$. This finishes the proof.
\end{proof}

\noindent While I-LSH \cite{Liu:2019} still generates less disk I/Os than \textit{roLSH-samp}, \textit{roLSH-samp} is significantly faster than I-LSH (due to less overall processing time) and also generates less disk seeks than I-LSH for bigger datasets (Section \ref{sec:exp}).

\begin{figure*}[h]
	\centering

	\begin{subfigure}[b]{0.31\textwidth}
		\centering
		{{\includegraphics[width=\linewidth]{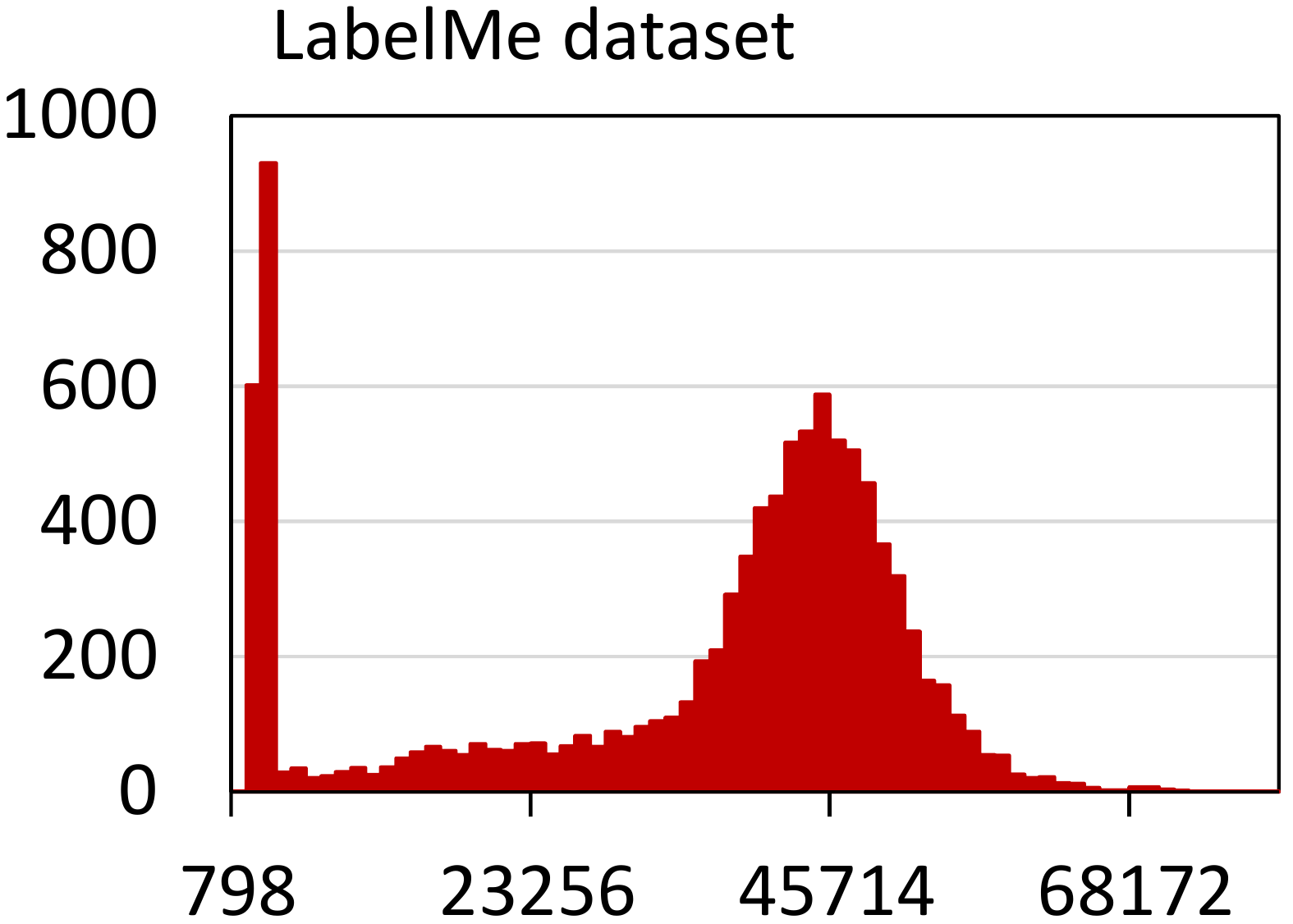}}}
	\end{subfigure}\quad
	\begin{subfigure}[b]{0.31\textwidth}
		\centering
		{{\includegraphics[width=\linewidth]{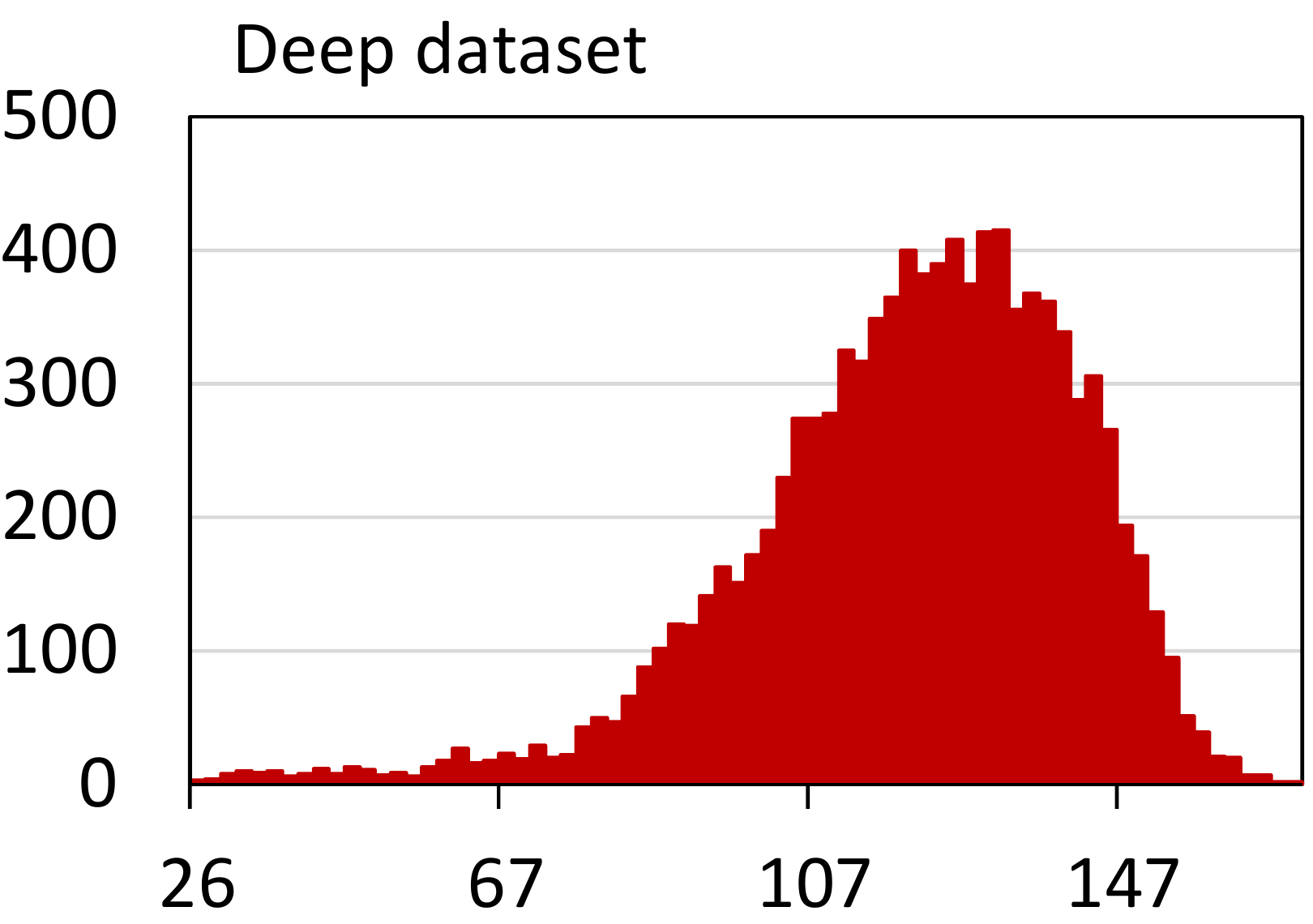}}}
	\end{subfigure}\quad
	\begin{subfigure}[b]{0.31\textwidth}
		\centering
		{{\includegraphics[width=\linewidth]{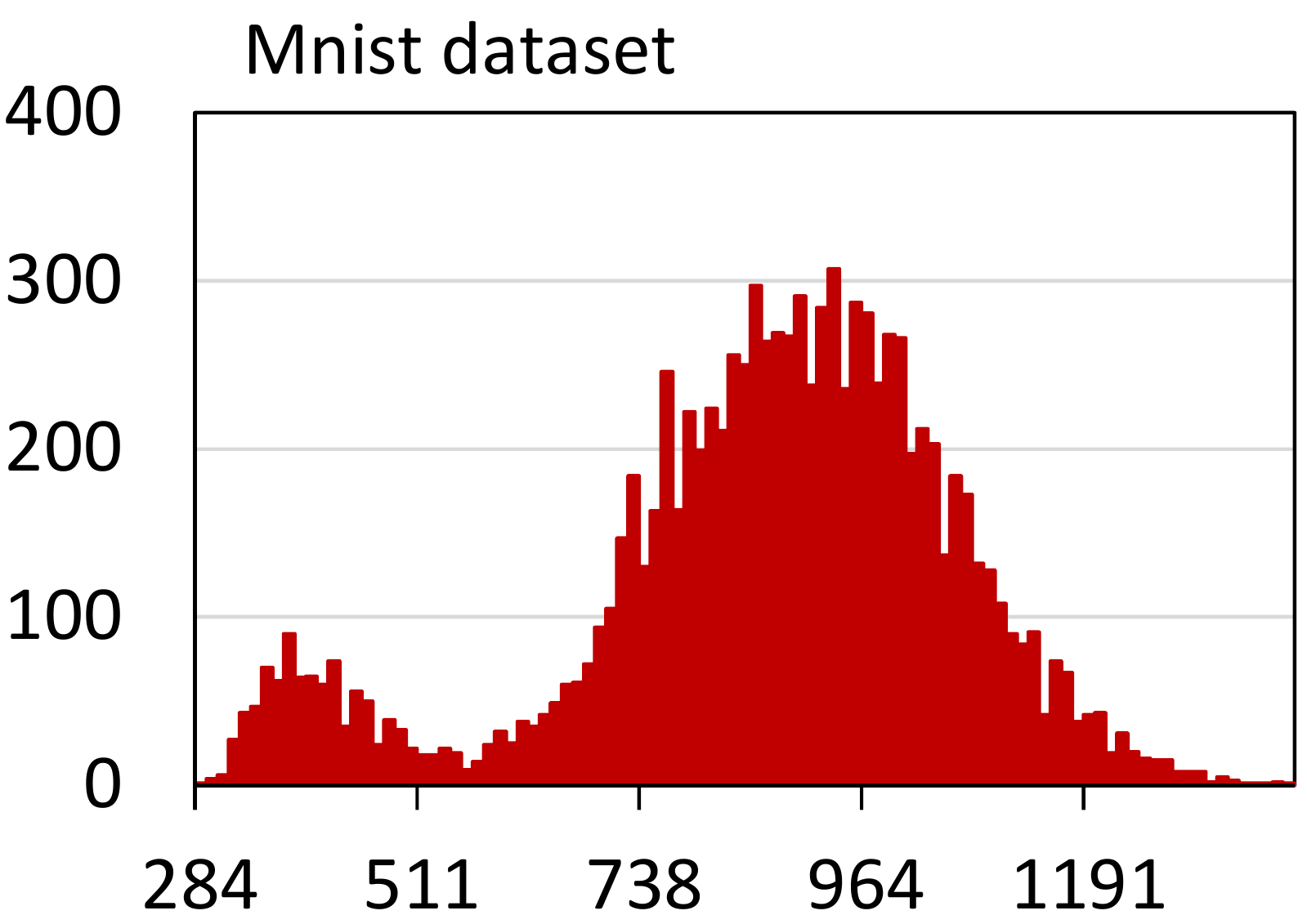}}}
	\end{subfigure}

	\caption{Frequency (Y Axis) of radiuses (X Axis) for 10,000 Top-100 Queries}
	\label{fig:radHists}
	\vspace*{-0.4in}
\end{figure*}

\subsection{Drawbacks of \textit{roLSH-samp}}
The main benefit of \textit{roLSH-samp} is that it is effective in reducing the disk I/Os, especially when the radiuses are large (e.g. the Audio dataset in Figure \ref{fig:diffSampling}). There is a minor overhead of utilizing the sampling-based method during the indexing phase. Additionally, we found that we also get good sampling representatives even with a small sampling size (e.g. 100). There are two main drawbacks of \textit{roLSH-samp}: 1) \textit{roLSH-samp} works best when Observation 1 holds true. We found out that Observation 1 holds true for many datasets, but not all. For example, as seen in Figure \ref{fig:radHists}, the radiuses for top-100 queries on the LabelMe dataset are quite different leading to inefficient performance of \textit{roLSH-samp} (as shown in Section \ref{sec:exp}), 2) It is not easy to do sampling for different $k$ values since the radius changes for different $k$ values. It is not trivial to build a single model and extend it to multiple $k$ values to find the radius for a particular $k$ value. Instead a model needs to be built for \textit{each} $k$ value. 

\subsection{Neural Network-based Prediction of Projected Radiuses}
\label{sec:rolsh-NN}

To remedy these two drawbacks, we present a Neural Network-based strategy, \textit{roLSH-NN}, that can better predict starting radiuses based on the query location (in each hash function) for any given $k$ value. The main intuition behind \textit{roLSH-NN} is that nearby points in the original space will have similar projected radiuses to find the desired number ($k$) of nearest neighbors. Hence, our goal is to predict the projected radiuses given the hash locations of a query for a given $k$.  

\noindent Formally, let $h_{i}(q)$ denote the bucket location of $q$ in the $i$th hash projection. Thus, $H(q) = h_{1}(q),...,h_{m}(q)$ denotes a vector of size $m$ (since there are $m$ hash projections) that contains $m$ bucket locations for a given query point $q$. Let $R_{act}(q, k)$ denote the smallest radius in the projected space that satisfies the desired number of results ($k$). Let $Q_{tr}$ be the set of training queries, where for each query $q\in Q_{tr}$, we also find out the ground truth (i.e. $R_{act}(q, k)$). This step is done in the indexing step, and hence does not affect the query processing time. We include and show this overhead in the indexing time in Section \ref{sec:exp}. We train a Neural Network with $Q_{tr}$ queries such that for each query $q$, we input $H(q)$ and $R_{act}(q, k)$ and the Neural Network outputs the predicted radius, $R_{pred}(q, k)$. We explain the different characteristics of our Neural Network in Section \ref{sec:exp}.  \\

\begin{table}[t] 
	\centerline{
		\begin{tabular}{|c|c|c|c|c|c|}\hline
			{\bf } & {\bf MLP} & {\bf Linear Reg.} & {\bf RANSAC} & {\bf Decision Tree} & {\bf Gradient Boosting}  \\ \hline\hline
			MSE & 0.0265 & 0.3543 & 0.3542 & 0.7057 & 0.2117 \\ \hline 
			R2 & 0.9687 & 0.5826 & 0.5827 & 0.1698 & 0.7504 \\ \hline 
		\end{tabular}
	}
	\caption{Performance Comparison of Learning Techniques}\label{tab:NN-Comparison}
	\vspace*{-0.3in}	
\end{table}

\noindent\textbf{Justification for choosing Neural Networks: }
Since the problem of predicting radiuses given the hash function is a regression problem, we tried several machine learning techniques. Table \ref{tab:NN-Comparison} shows that Neural Networks (denoted by \textit{MLP} since we use a Multilayer Perceptron Neural Network) have the best MSE and R2 for a sample dataset (Deep) for $Q_{tr}=10,000$ among different machine learning techniques (using 10-fold cross validation). Hence, we choose Neural Networks over other techniques in the design of \textit{roLSH-NN}. \\ 

\noindent\textbf{Underestimation of Radius: }When the radius is underestimated (i.e. $R_{pred}(q, k)$ $< R_{act}(q, k)$), the desired number of results are not found and hence we have to enlarge the radius in all projections. One strategy is to follow the same expansion pattern of \textit{roLSH-samp} presented in Section \ref{sec:roLSH-samp}, where the predicted radius is set as $i2R$. We call this strategy \textit{roLSH-NN-iVR}. The drawback of this strategy is that it can lead to excessive (and expensive) disk seeks if the predicted radius is much lower than the actual radius. Since we observe that $R_{pred}(q, k)$ is close to $R_{act}(q, k)$, we also adopt another strategy where we increase the predicted radius, $R_{pred}(q, k)$, linearly by $R_{inc}$ such that $R_{inc} = R_{pred}(q, k) \times \lambda$. This strategy is referred to as \textit{roLSH-NN-$\lambda$} in the rest of this paper. \\

\noindent\textbf{Overestimation of Radius: }While overestimation of the projected radius by the Neural Network leads to wasted disk I/Os during query processing, we experimentally show in Section \ref{sec:exp} that these wasted disk I/Os are still less than the exponential strategy of C2LSH/QALSH and the improvement in the query processing time (as compared with I-LSH) offsets the disk I/Os significantly. 

\noindent\textbf{Extension to any $k$: }In order to train the Neural Network to work for any number of desired results ($k$), we need to include $k$ as an input feature in the training set. In order to simplify the training procedure, we only consider few values of $k$ in the training set $Q_{tr}$. In Section \ref{sec:exp}, we explain the training setup and the effect of using different $k$ values during the training on the  
time.

\section{Experimental Evaluation}
\label{sec:exp}
In this section, we evaluate the effectiveness of our proposed index structure, \textit{roLSH}, on three real diverse high-dimensional datasets. All experiments were run on the nodes of the Bigdat cluster \footnote{Supported by NSF Award \#1337884} with the following specifications: two Intel Xeon E5-2695, 256GB RAM, and CentOS 6.5 operating system. We implement our work on top of C2LSH \cite{Gan:2012:LHS:2213836.2213898} since we found it to be the fastest external memory-based LSH algorithm (while achieving high accuracy for high-dimensional datasets). Note that, our method is orthogonal to the LSH algorithm and can be used in any state-of-the-art LSH algorithms. 
All codes were written in C++11 and compiled with gcc v4.7.2 with the -O3 optimization flag.
We compare our three strategies, \textit{roLSH-samp}, \textit{roLSH-NN-iVR}, and \textit{roLSH-NN-$\lambda$} with the state-of-the-art LSH algorithms C2LSH \cite{Gan:2012:LHS:2213836.2213898} and I-LSH \cite{Liu:2019}. 

\begin{table}
	\centerline{
		\begin{tabular}{|c|c|c|c|}\hline
			{\bf Index Size} & {\bf LabelMe} & {\bf Deep} & {\bf Mnist}\\ \hline\hline
			roLSH-samp & 164.1 & 744.1 & 2120.1 \\ \hline 
			roLSH-NN & 164.4 & 744.5 & 2120.5 \\ \hline 
			C2LSH & 164 & 744 & 2120 \\ \hline 
			I-LSH & 81 & 648 & 5265 \\ \hline 
		\end{tabular}
		\quad
		\begin{tabular}{|c|c|c|c|}\hline
			{\bf Index Time} & {\bf LabelMe} & {\bf Deep} & {\bf Mnist}\\ \hline\hline
			roLSH-samp & 83.5 & 93.5 & 1480.8 \\ \hline 
			roLSH-NN & 88.6 & 98.4 & 1488 \\ \hline 
			C2LSH & 80.6 & 69 & 1430.2 \\ \hline 
			I-LSH & 20.8 & 25.7 & 1359.9 \\ \hline 
		\end{tabular}	
	}

	\caption{Comparison of (a) Index Size (in MB) and (b) Index Construction Time (in sec) on Different Datasets}\label{tab:IndexingComp}

\end{table}

\begin{figure*}
	\centering
	\begin{subfigure}[b]{0.31\textwidth}
		\centering
		{\includegraphics[width=\linewidth]{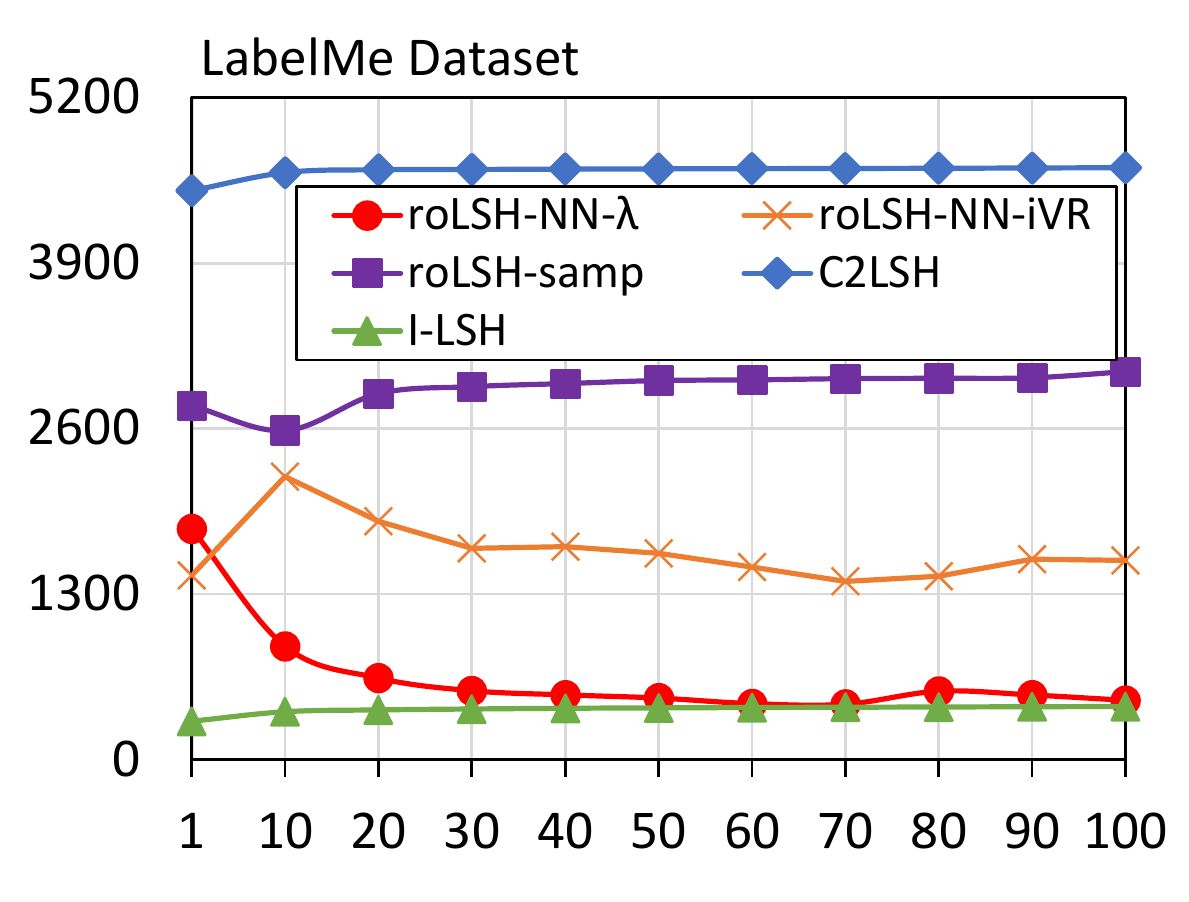}}
	\end{subfigure}\quad
	\begin{subfigure}[b]{0.31\textwidth}
		\centering
		{{\includegraphics[width=\linewidth]{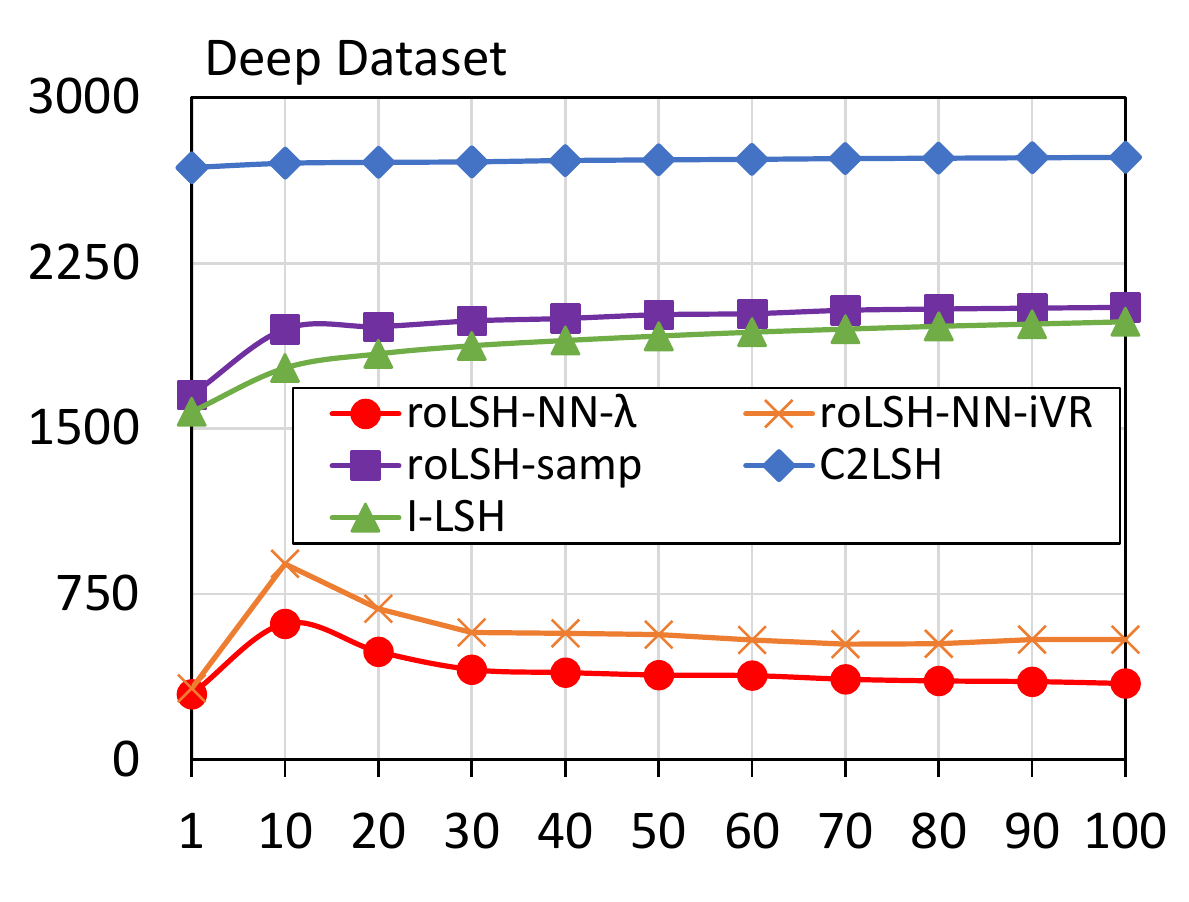}}}
	\end{subfigure}\quad
	\begin{subfigure}[b]{0.31\textwidth}
		\centering
		{{\includegraphics[width=\linewidth]{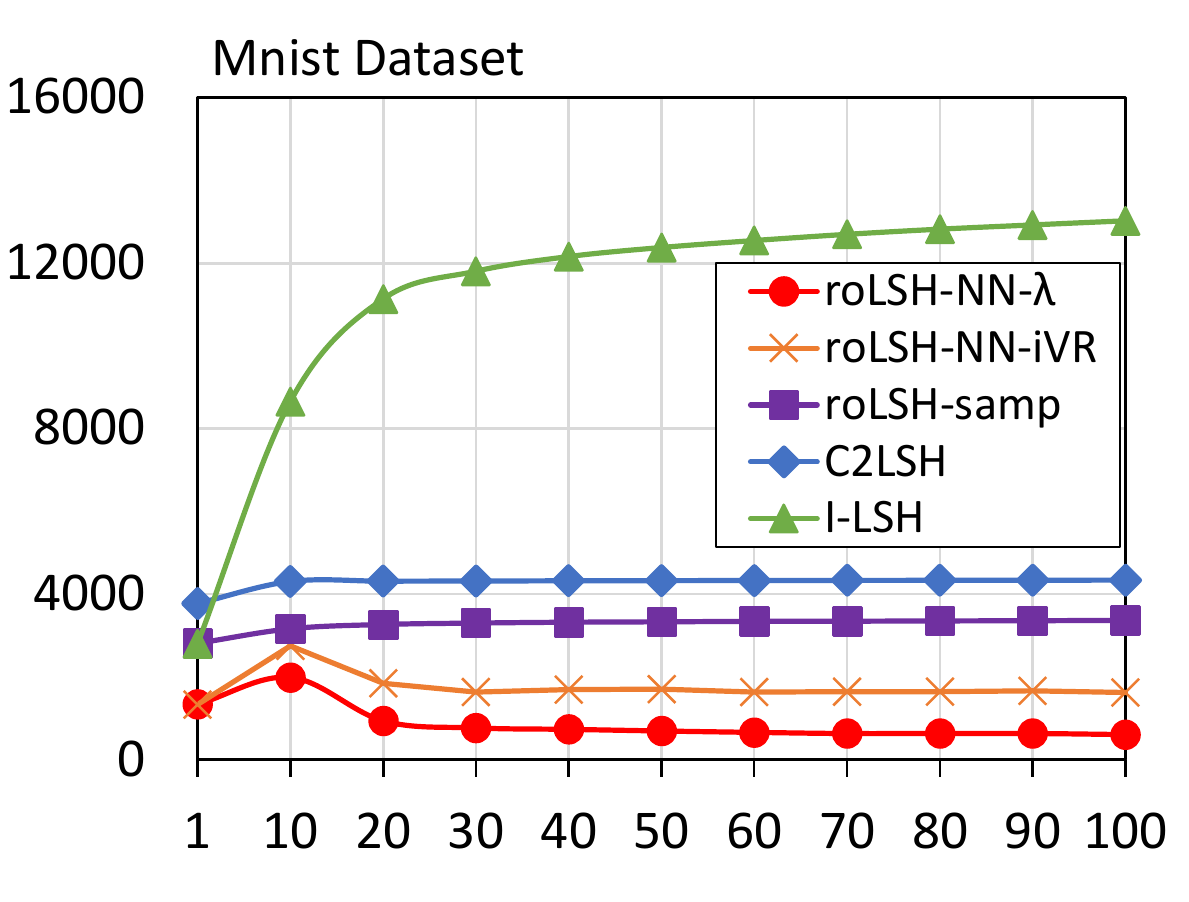}}}
	\end{subfigure}
	\caption{Number of Disk Seeks (Y axis) for different $k$ (X Axis) on 3 datasets}
	\label{fig:expdiskseek}

\end{figure*}

\begin{figure*}
	\centering
	\begin{subfigure}[b]{0.31\textwidth}
		\centering
		{\includegraphics[width=\linewidth]{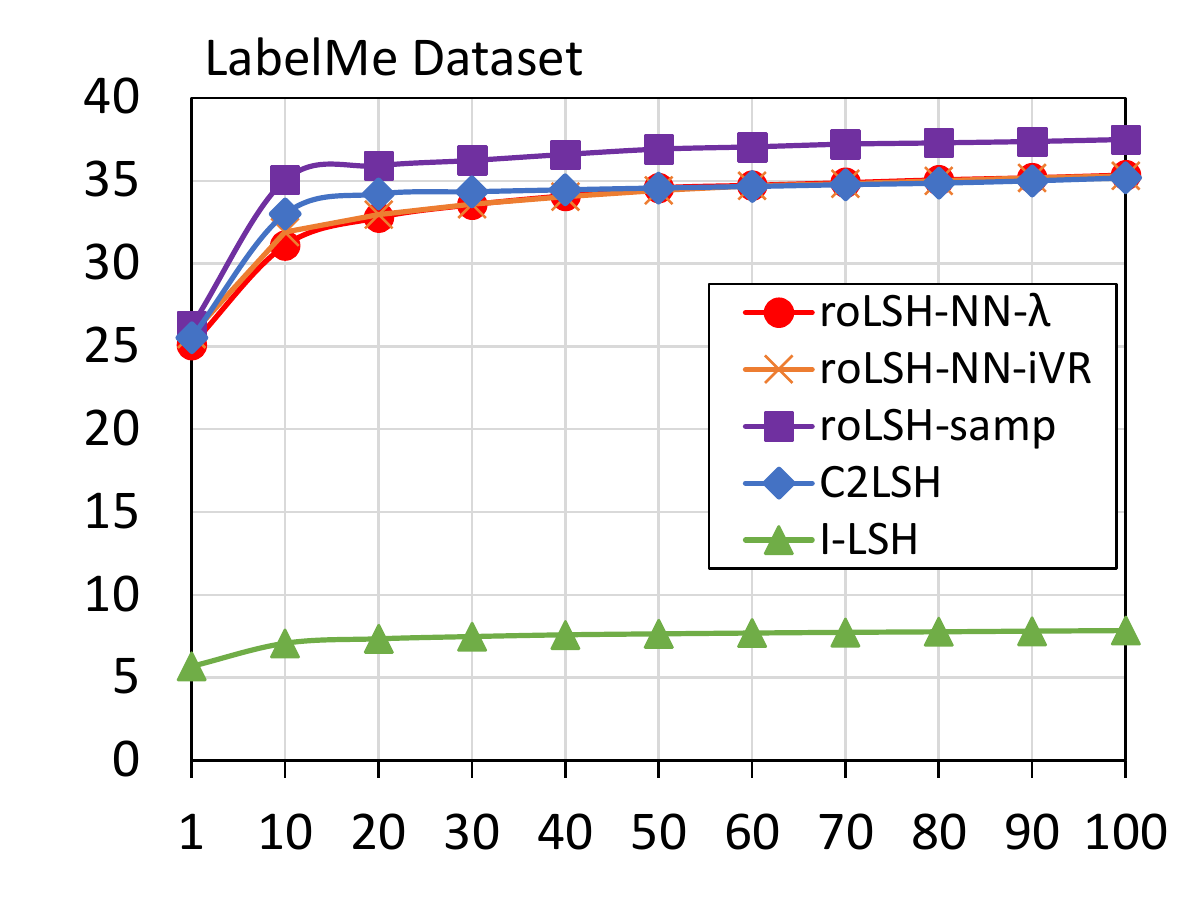}}
	\end{subfigure}\quad
	\begin{subfigure}[b]{0.31\textwidth}
		\centering
		{{\includegraphics[width=\linewidth]{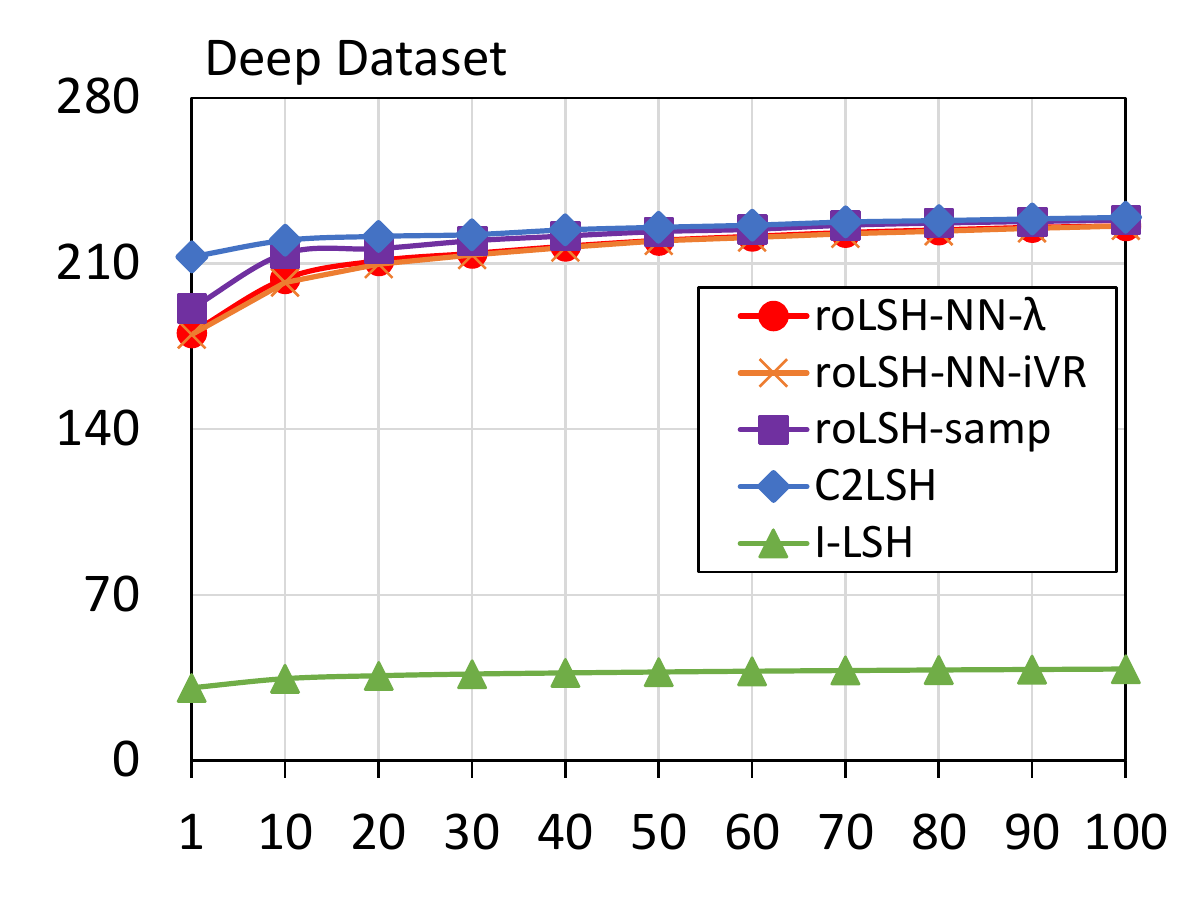}}}
	\end{subfigure}\quad
	\begin{subfigure}[b]{0.31\textwidth}
		\centering
		{{\includegraphics[width=\linewidth]{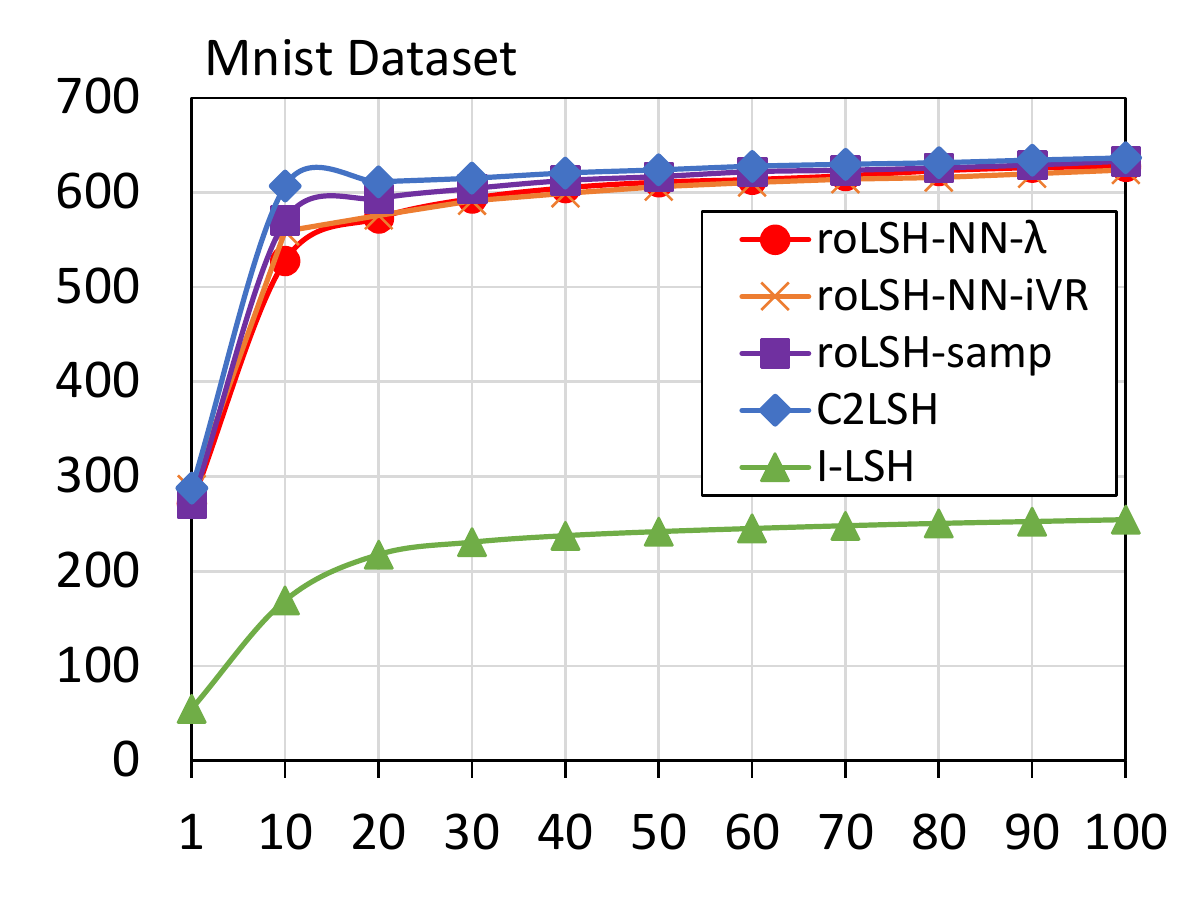}}}
	\end{subfigure}

	\caption{Amount of Data Read (in MB) (Y axis) for $k$ (X Axis) on 3 datasets}
	\label{fig:expdiskIO}

\end{figure*}

\begin{figure*}
	\centering
	\begin{subfigure}[b]{0.31\textwidth}
		\centering
		{\includegraphics[width=\linewidth]{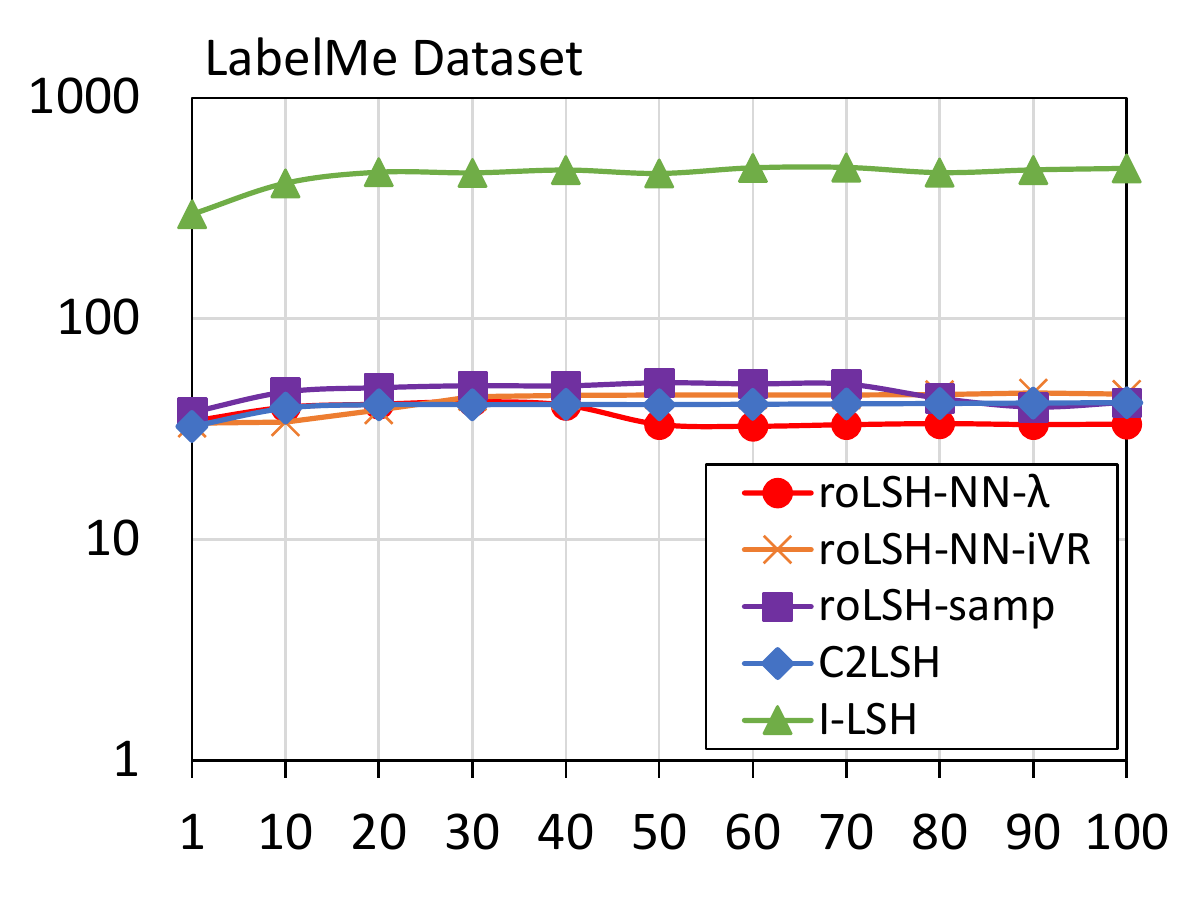}}
	\end{subfigure}\quad
	\begin{subfigure}[b]{0.31\textwidth}
		\centering
		{{\includegraphics[width=\linewidth]{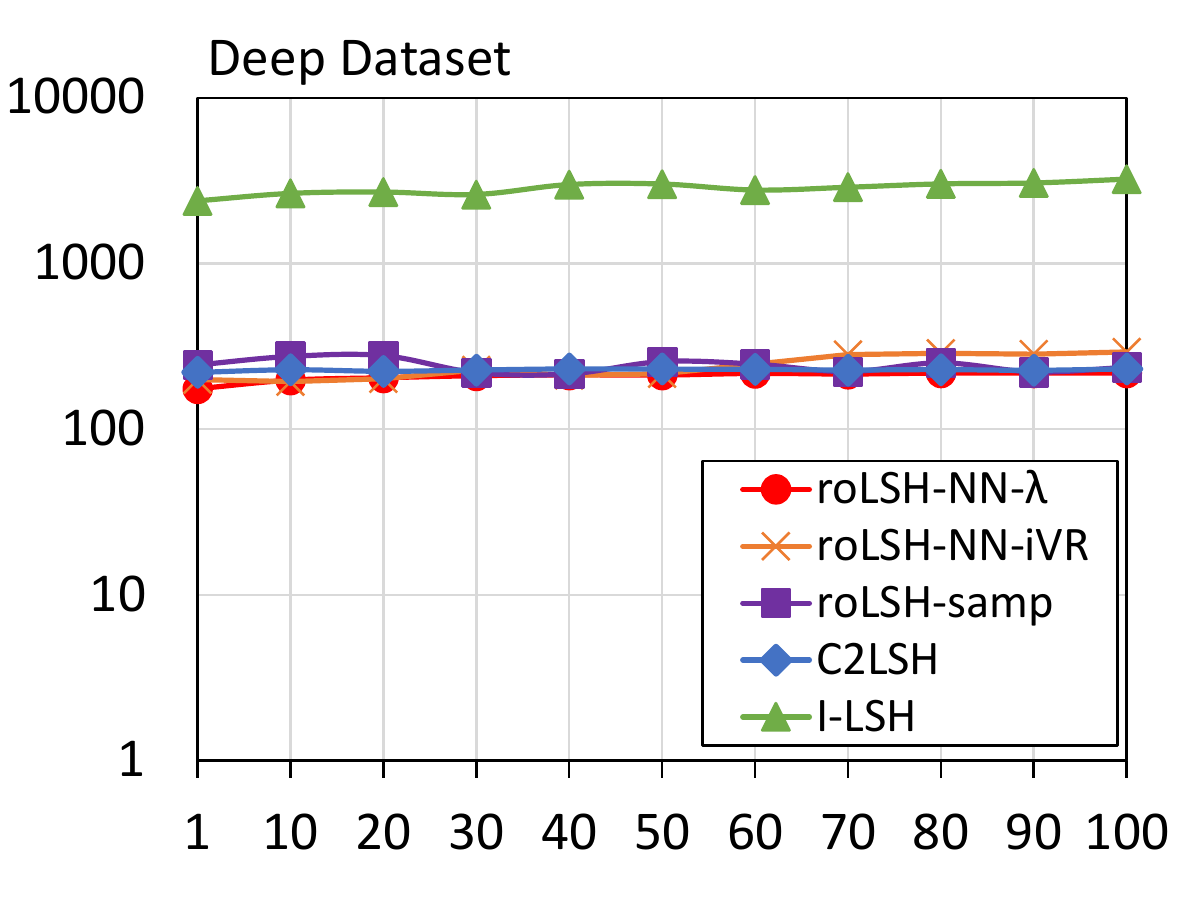}}}
	\end{subfigure}\quad
	\begin{subfigure}[b]{0.31\textwidth}
		\centering
		{{\includegraphics[width=\linewidth]{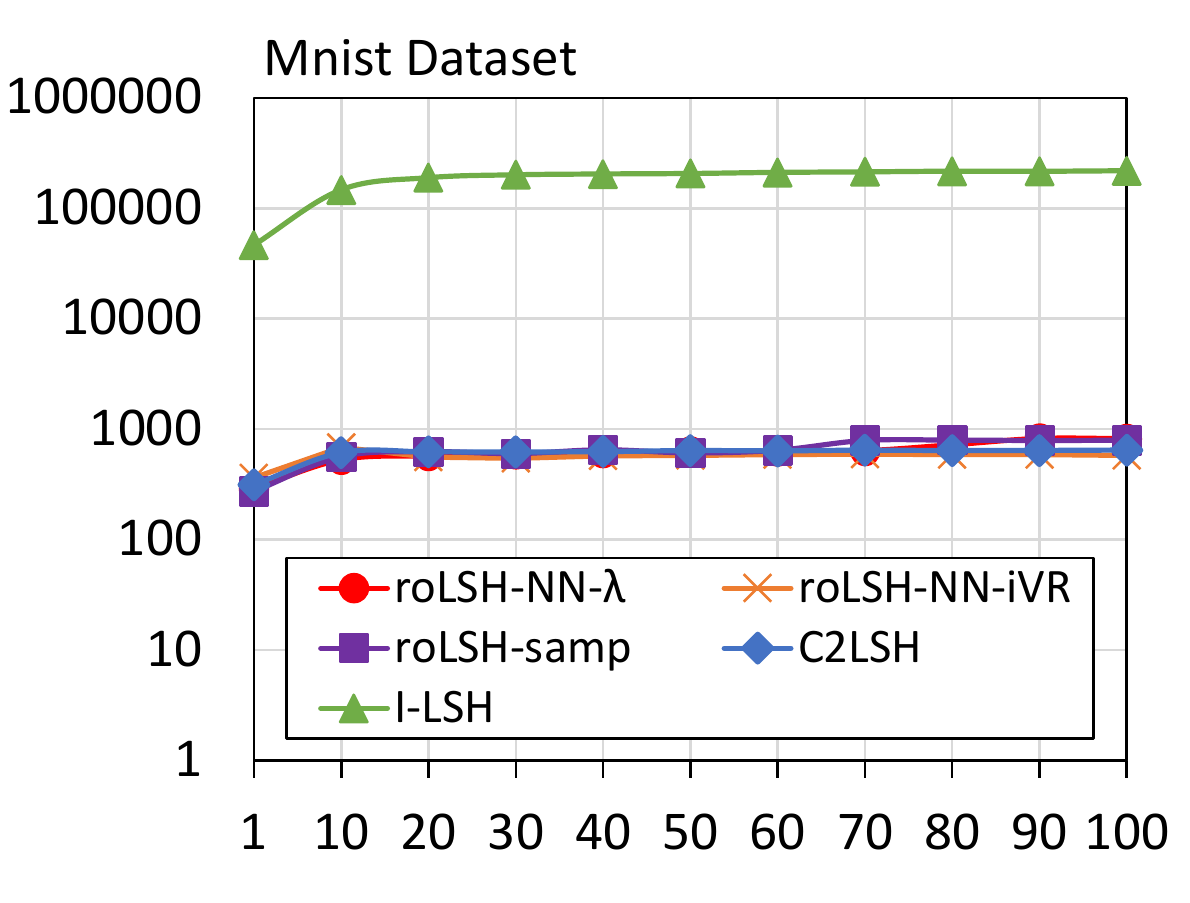}}}
	\end{subfigure}

	\caption{Algorithm Time (in ms, \textit{log scale}) (Y axis) for $k$ (X Axis) on 3 datasets}
	\label{fig:expAlgTime}
\end{figure*}

\begin{figure*}
	\centering
	\begin{subfigure}[b]{0.31\textwidth}
		\centering
		{\includegraphics[width=\linewidth]{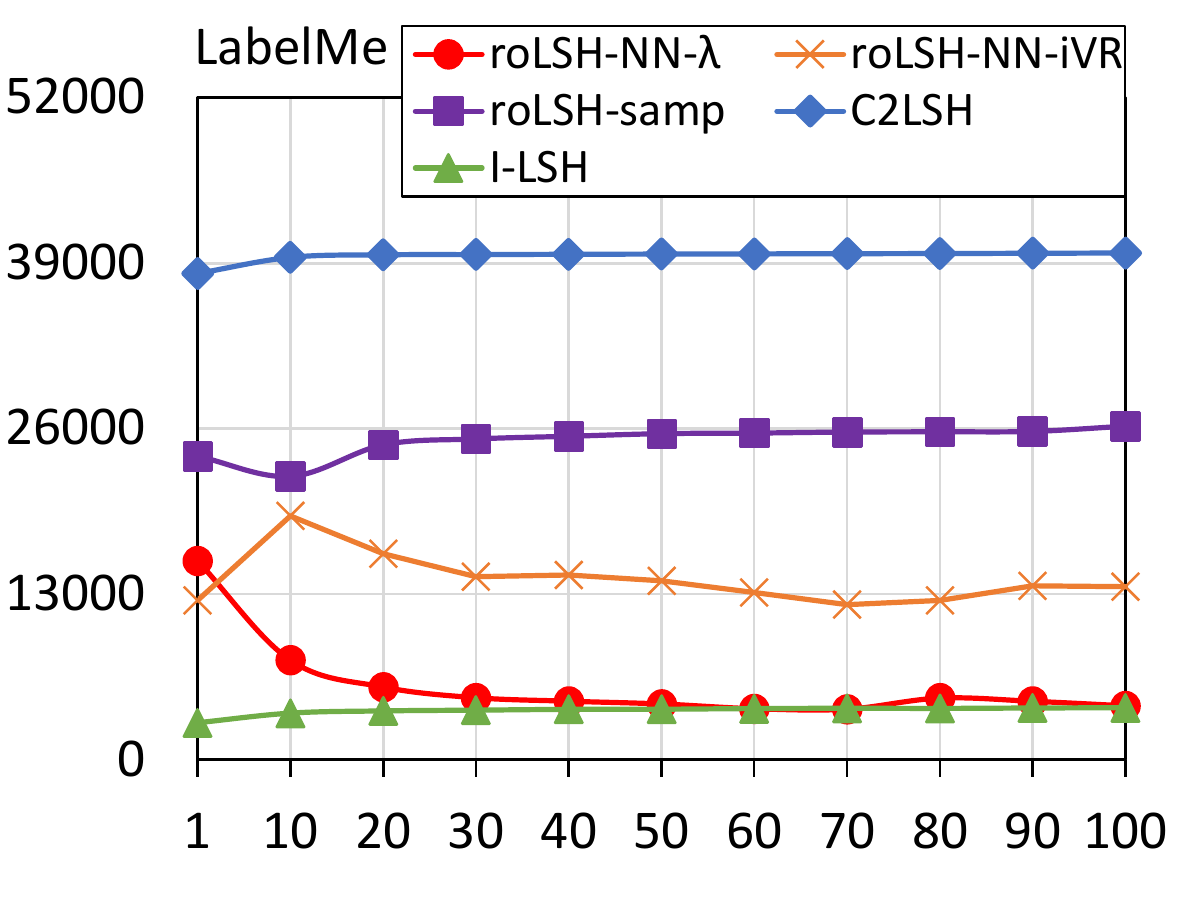}}
	\end{subfigure}\quad
	\begin{subfigure}[b]{0.31\textwidth}
		\centering
		{{\includegraphics[width=\linewidth]{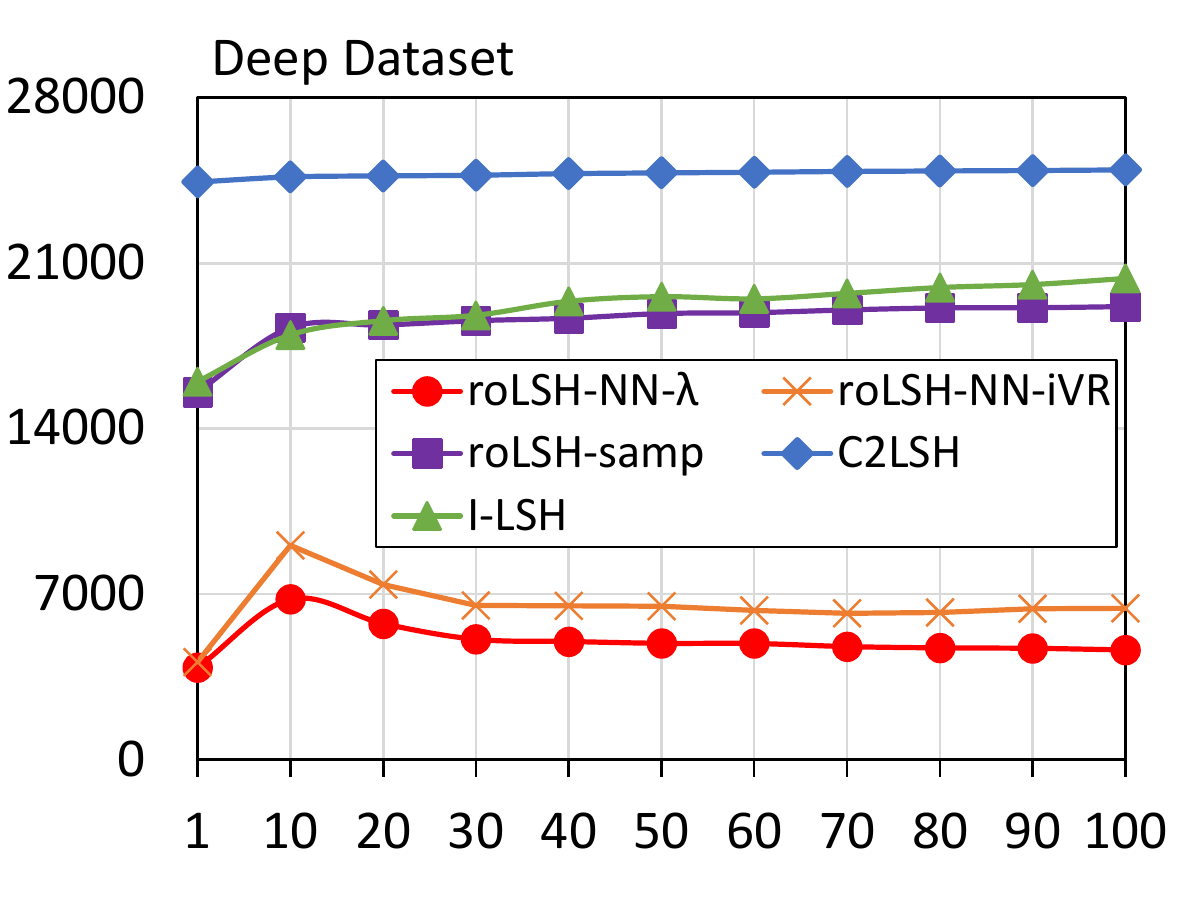}}}
	\end{subfigure}\quad
	\begin{subfigure}[b]{0.31\textwidth}
		\centering
		{{\includegraphics[width=\linewidth]{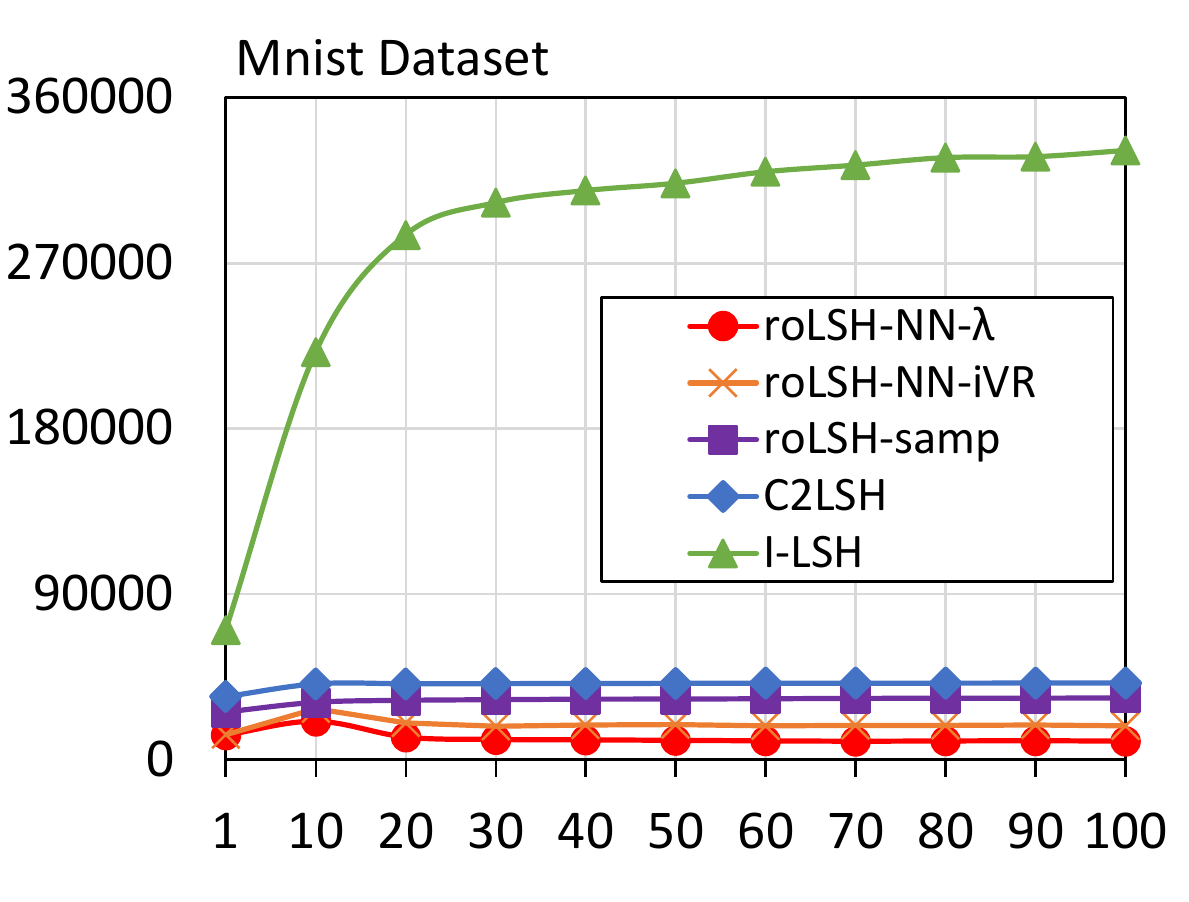}}}
	\end{subfigure}

	\caption{Query Processing Time (in ms) (Y axis) for $k$ (X Axis) on 3 datasets}
	\label{fig:expTime}
\end{figure*}

\begin{figure*}
	\centering
	\begin{subfigure}[b]{0.31\textwidth}
		\centering
		{\includegraphics[width=\linewidth]{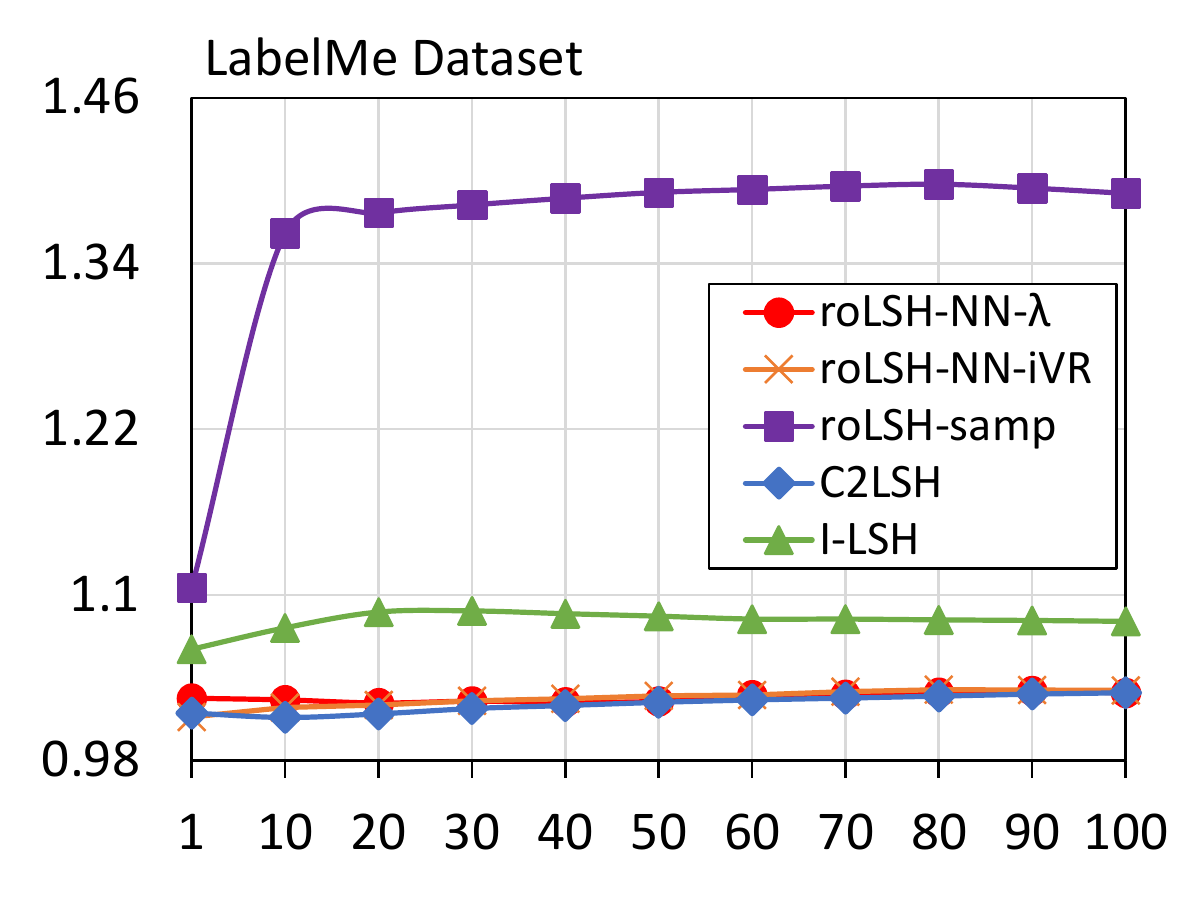}}
	\end{subfigure}\quad
	\begin{subfigure}[b]{0.31\textwidth}
		\centering
		{{\includegraphics[width=\linewidth]{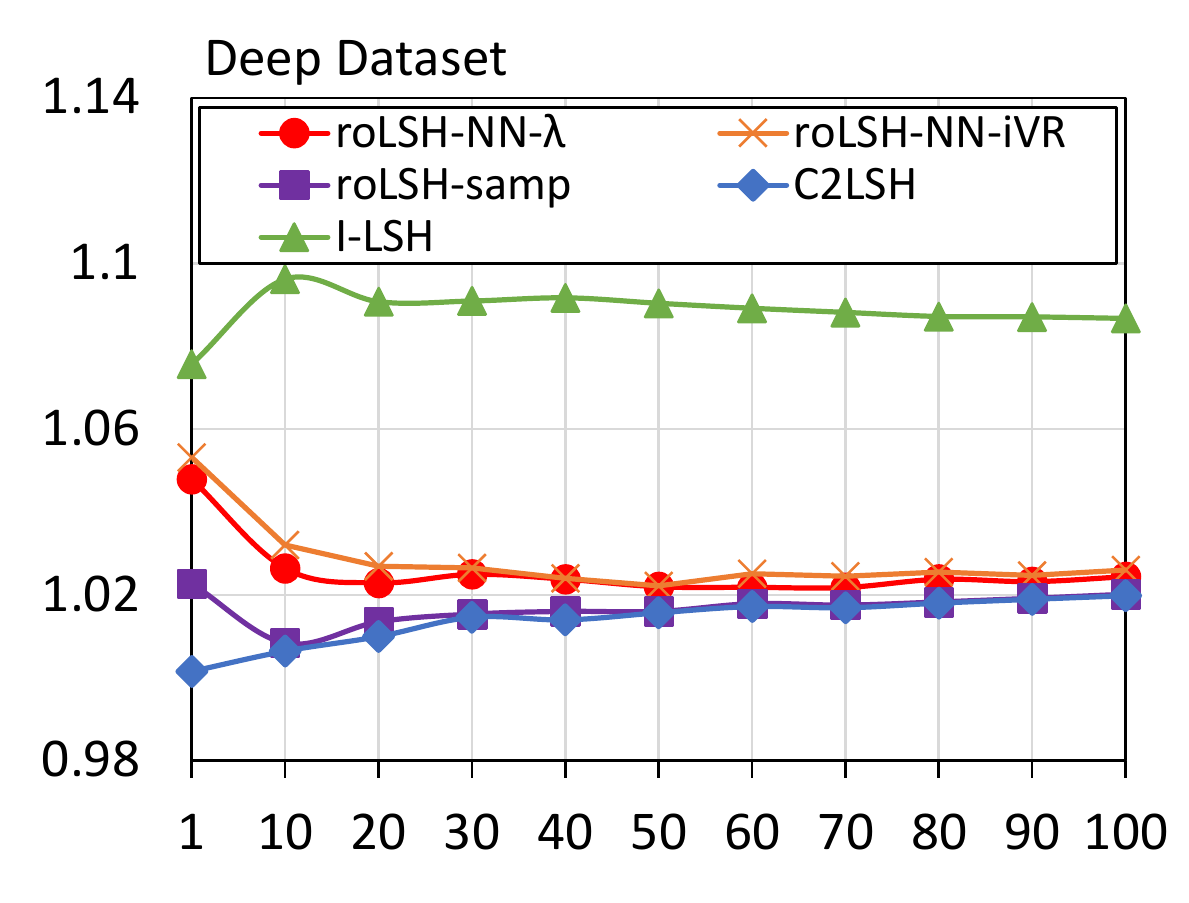}}}
	\end{subfigure}\quad
	\begin{subfigure}[b]{0.31\textwidth}
		\centering
		{{\includegraphics[width=\linewidth]{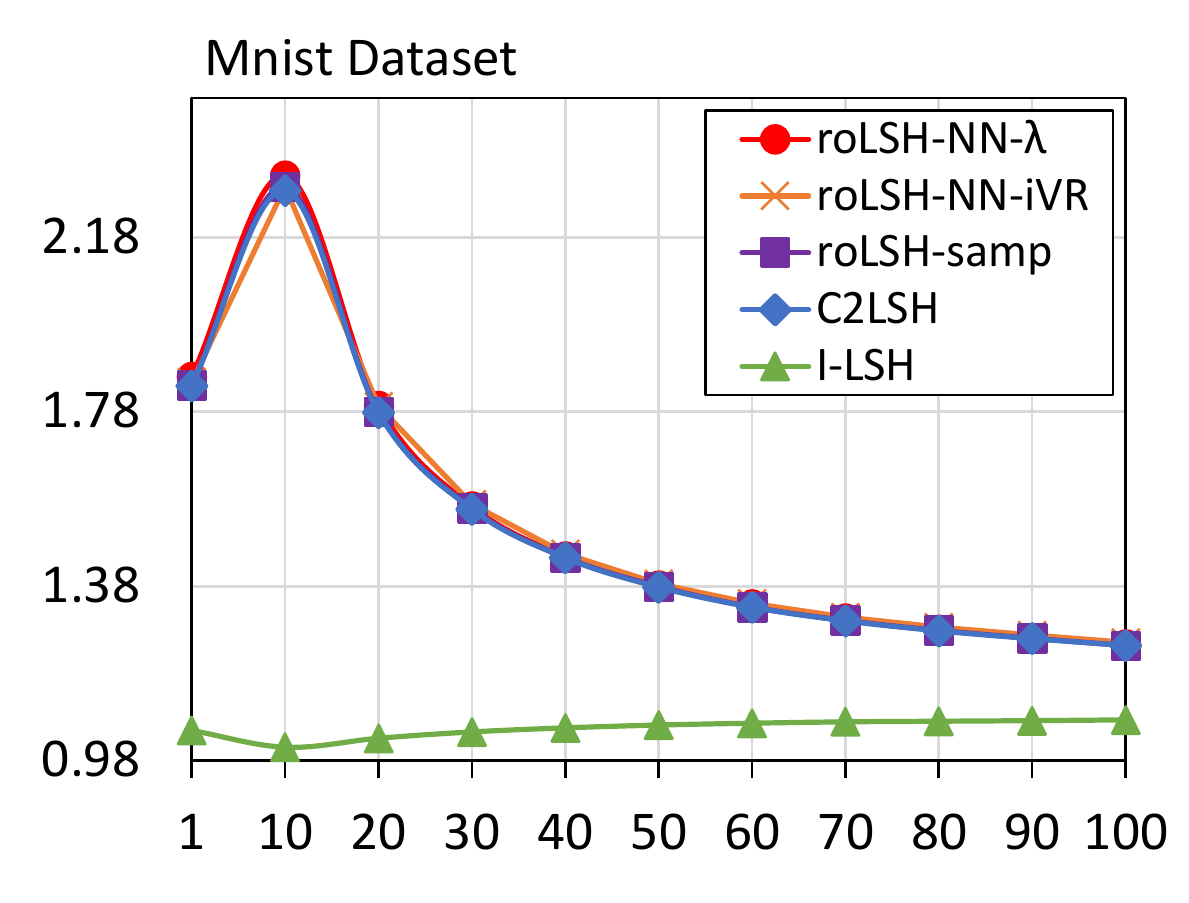}}}
	\end{subfigure}

	\caption{Accuracy Ratio (Y axis) for different $k$ (X Axis) on 3 datasets}
	\label{fig:expRatio}
\end{figure*}

\subsection{Datasets}
We use the following three popular real datasets to evaluate the proposed method. These datasets cover different sizes and are enough to show the scalability of the different strategies. The details of these datasets are as follows:

\begin{itemize}
	\item \textbf{LabelMe}\cite{russell2008labelme} consists of $181,093$ 512-dimensional points which were generated by running the GIST feature extraction algorithm on annotated images.

	\item \textbf{Deep} consists of $1,000,000$ 96-dimensional points that were randomly chosen from the Deep1B dataset introduced in \cite{babenko2016efficient}. 
	
	\item \textbf{Mnist}\cite{loosli2007training} This dataset contains $8,100,000$ 784-dimensional points that represent images of the digits 0 to 9 which are grayscale and of size 28 $\times$ 28.
\end{itemize}

\subsection{Evaluation Criteria and Parameters}

The goal of \textit{roLSH} is to improve the performance efficiency without sacrificing the accuracy of existing LSH techniques. 
The performance and accuracy of the technique used in this paper are evaluated using the following metrics: 

\begin{itemize}

	\item \textbf{Query Processing Time ($QPT$):} We break down the Query Processing Time into the Index I/O cost, the Algorithm time ($AlgTime$), and the negligible false positive removal cost (denoted by $FPRemTime$, which consists of the cost of reading the data point candidates and computing their exact Euclidean distance for removing false positives). Following \cite{Liu:2019}, we further break down Index I/O cost into the number of disk seeks (i.e. random I/O reads, $noDiskSeeks$) and the amount of data (i.e. index files, $dataRead$) read in MB. Following \cite{Seagate}, for a Seagate 1TB HDD with 7200 RPM, we assume a random seek to cost 8.5 ms on average, and the average time to read data to be 0.156 MB/ms. Thus, we have $QPT = noDiskSeeks*8.5 + dataRead*0.156 + AlgTime + FPRemTime$. 
	
	\item \textbf{Accuracy:} We follow the accuracy ratio definition followed by many previous works (\cite{Liu:2014:SEI:2732939.2732947,Gan:2012:LHS:2213836.2213898,Huang:2015:QLH:2850469.2850470}): $\frac{1}{k}\sum_{i=1}^{k}\frac{||o_i, q||}{||o_i^*,q||}$. Here, $o_i$ is the $i$th point returned by the technique and $o_i^*$ is the true $i$th nearest point from $q$ (ground truth). Ratio of 1 means the returned results have the same distance from the query as the ground truth. The closer the ratio is to 1, the higher is the accuracy. 
\end{itemize}

For the state-of-art methods, we used the same parameters suggested in their papers ($w=2.719$ for QALSH and $w=2.184$ for C2LSH). Also, as \textit{roLSH} is built on top of C2LSH, it uses the same parameters as C2LSH. We set the allowed error probability, $\delta$, to be 0.1.  
The Multilayer Perceptron (MLP) Neural Network is implemented using the Scikit-learn Python package \cite{scikit-learn}. In this paper, we use the default parameters and options (i.e. 100 hidden layers, \textit{ReLU} activation function, and the \textit{Adam} optimization algorithm). We leave the hyper-parameter tuning analysis to future work. We choose 10,000 training queries randomly from the dataset. 50 different queries were randomly chosen from the dataset for the evaluation. We report an average of the results on these 50 queries.

\subsection{Effect of Different Parameters on Performance of \textit{roLSH-NN-$\lambda$}}
\label{sec:exprolshNN}

In this section, we present the performance of \textit{roLSH-NN-$\lambda$} under different parameters for the Deep dataset.\\ 
\noindent\textbf{Effect of Training Size:} We consider three different training sizes (5K, 10K, 50K). In our experiments, the MSE reduces (by 18.4\% between 5K and 50K training size) as the training size increases and since the MSE decreases (i.e. the predicted radius is close to the actual radius), the overall Query Processing Time (QPT) also decreases (by 33\% between 5K and 50K training size). In the following experiments, we choose 10K as the default training size. Due to space limitations, we do not present the results in this Section \ref{sec:exprolshNN} in detail.

\noindent\textbf{Effect of Number of Different $k$ in Training:} We analyze the performance of \textit{roLSH-NN-$\lambda$} for different values of k that are present in the training data while keeping the total training size and $\lambda$ constant. We chose \{1, 50, 100\}, \{1, 25, 50, 75, 100\}, and \{1, 10, 25, 50, 75, 90, 100\} as three different settings. The MSE reduces as more diverse $k$ are included: by 33\% between the first two settings, but only by 14\% between the last two settings since the neural networks are capable of adequately predicting the radiuse for different $k$ even for the second setting (which is our default in the following experiments).

\noindent\textbf{Effect of Different Radius Increment ($\lambda$):}  
We experiment using $\lambda$ values of 5\%, 10\%, and 20\%. 
As $\lambda$ increases, the number of disk seeks decrease (by 32\%) since a higher $\lambda$ eventually results in a larger radius and in turn makes the algorithm stop sooner without processing all projections, but the algorithm time and the amount of I/O increases (by 4\% and 1\% respectively) since more hash buckets are processed. We choose 10\% as our default in further experiments.

\subsection{Discussion of the Results}

Table \ref{tab:IndexingComp} (a) shows the index sizes of all techniques on all datasets. Since we use C2LSH as our underlying LSH implementation, the sizes of \textit{roLSH} are similar to that of C2LSH. The reported size of \textit{roLSH-NN} includes the Neural Network model, and hence is very slightly higher than C2LSH. Table \ref{tab:IndexingComp} (b) shows the time taken to finish the index construction. The reported times show that the sampling and training overhead for \textit{roLSH-samp} and \textit{roLSH-NN} are only 3.4\% and 3.9\% for the largest dataset (Mnist). 

\noindent\textbf{Number of Disk Seeks: }Figure \ref{fig:expdiskseek} shows the number of disk seeks (random I/Os) required by these different techniques. It is very interesting to note that while I-LSH performs the best (\textit{roLSH-NN-$\lambda$} is a close second) for LabelMe, their performance degrades as the dataset size increases. I-LSH produces significant more disk seeks as the dataset size increases. We believe this is mainly due to the fact that more points need to be accessed incrementally to find the candidates. \textit{roLSH-NN-$\lambda$} significantly performs the best for Deep and Mnist datasets because it can accurately predict the radius for different $k$. Every time \textit{roLSH-NN-$\lambda$} underestimates the radius (Section \ref{sec:rolsh-NN}), it has to increment the radius by $\lambda$ resulting in a disk seek in each projection. Also, as expected, \textit{roLSH-NN-iVR} produces more disk seeks due to radius underestimation (Section \ref{sec:rolsh-NN}).

\noindent\textbf{Amount of Data Read: }Figure \ref{fig:expdiskIO} shows the total amount of data (index files) read. Since I-LSH incrementally increases the search to the nearest point in the projected space (instead of an empirically chosen number, such as $\lambda$), it results in the least amount of data read for all datasets. These savings in the I/O are offset due to the expensive search for the nearest point as shown in Figure \ref{fig:expAlgTime}. Especially for lower $k$, \textit{roLSH-NN-iVR} and \textit{roLSH-NN-$\lambda$} read less data than C2LSH, but as $k$ increases the overall data read is similar for both techniques. It is interesting to note that \textit{roLSH-samp} reads significantly more data for LabelMe dataset. This is due to choosing of a bad starting radius due to the unique distribution of the LabelMe radiuses (Figure \ref{fig:radHists} (a)). Moreover, \textit{roLSH-NN-iVR} and \textit{roLSH-NN-$\lambda$} read similar amount of data since their starting radius is the same.

\noindent\textbf{Algorithm Time: }Figure \ref{fig:expAlgTime} shows the time needed by the algorithms to find the candidates (excluding the time taken to read the index files). Note the log scale of this figure because the algorithm time for I-LSH was orders of magnitude more than the other techniques. This is because I-LSH expands the radius incrementally in each projection which creates a significant overhead. Figure \ref{fig:expAlgTime} shows that the overhead of our methods is negligible when compared with C2LSH.

\noindent\textbf{Query Processing Time: }Figure \ref{fig:expTime} shows the overall time required to solve a given k-NN query. I-LSH works well for smaller datasets (LabelMe) but is significantly slower as the dataset size increases (due to high overhead in incrementally finding the next neighbor in each projection). \textit{roLSH-samp} is always faster than C2LSH because of the savings in disk seeks. \textit{roLSH-NN-iVR} and \textit{roLSH-NN-$\lambda$} are always much faster than \textit{roLSH-samp} and C2LSH because of their ability to accurately predict radiuses, resulting in significantly less disk seeks and lesser (or similar in some cases) data read than C2LSH. \textit{roLSH-NN-$\lambda$} has better performance compared to \textit{roLSH-NN-iVR}, mainly because of having lesser disk seeks as discussed before. This figure shows the performance benefit of \textit{roLSH-NN-$\lambda$} over its competitors for different datasets, and confirms that the design of \textit{roLSH-NN-$\lambda$} leads to improvement in overall efficiency. 

\noindent\textbf{Accuracy: }Figure \ref{fig:expRatio} shows the accuracy of all techniques. \textit{roLSH-samp} gives the worst accuracy for LabelMe dataset. We found that this is due to the fact that LabelMe dataset has queries with very different large radiuses. \textit{roLSH-samp} is unable to work well for datasets that have differing radiuses because if the starting radius is chosen wrong, then \textit{roLSH-samp} can significantly overestimate the radius for larger radiuses leading to lower accuracy. \textit{roLSH-NN-$\lambda$} always returns similar accuracy to that of C2LSH. I-LSH returns a better accuracy for Mnist dataset due to their usage of \textit{query-aware} hash functions, but the performance is significantly slower as shown in Figure \ref{fig:expTime}.

\section{Conclusion}

Locality Sensitive Hashing is a popular technique for efficiently solving Approximate Nearest Neighbor queries in high-dimensional spaces. State-of-the-art LSH techniques improve the overall disk I/Os at the expense of algorithm time. 
In this paper, we present a unique index structure called radius-optimized Locality Sensitive Hashing (\textit{roLSH}). The goal of \textit{roLSH} is to improve the efficiency of LSH techniques by improving the random disk seeks without any significant overhead in algorithm time. We propose two novel strategies, \textit{roLSH-samp} and \textit{roLSH-NN} that are based on sampling and Neural Networks respectively. Experimental results on real datasets show the benefit of \textit{roLSH} in improving overall performance over existing state-of-the-art techniques, C2LSH and I-LSH. 

\end{document}